%% file: LowAdaptivityUSM.tex
\documentclass[11pt]{article}
\usepackage{graphicx}
\usepackage{latexsym}
\usepackage{amssymb}
\usepackage{amsmath}
\usepackage{amsthm}
\usepackage{forloop}
\usepackage[paper=letterpaper,margin=1in]{geometry}
\usepackage[ruled,vlined,linesnumbered]{algorithm2e}
\usepackage{nicefrac}
\usepackage{paralist}
\usepackage{xspace}
\usepackage[capitalize]{cleveref}
\usepackage{breqn}
\usepackage{xcolor}
\usepackage{mleftright}

\newtheorem{theorem}{Theorem}[section]
\newtheorem{lemma}[theorem]{Lemma}

\newtheorem{corollary}[theorem]{Corollary}

\newtheorem{proposition}[theorem]{Proposition}
\newtheorem{observation}[theorem]{Observation}

\makeatletter
\newtheorem*{rep@theorem}{\rep@title}
\newcommand{\newreptheorem}[2]{%
\newenvironment{rep#1}[1]{%
 \def\rep@title{#2 \ref{##1}}%
 \begin{rep@theorem}}%
 {\end{rep@theorem}}}
\makeatother

\newreptheorem{theorem}{Theorem}
\newreptheorem{reduction}{Reduction}
\newreptheorem{proposition}{Proposition}
\newreptheorem{lemma}{Lemma}

\newcommand{\defcal}[1]{\expandafter\newcommand\csname c#1\endcsname{{\mathcal{#1}}}}
\newcommand{\defbb}[1]{\expandafter\newcommand\csname b#1\endcsname{{\mathbb{#1}}}}
\newcounter{calBbCounter}
\forLoop{1}{26}{calBbCounter}{
    \edef\letter{\Alph{calBbCounter}}
    \expandafter\defcal\letter
		\expandafter\defbb\letter
}

\newcommand{\eps}{\varepsilon}
\newcommand{\ie}{{\it i.e.}}
\newcommand{\eg}{{\it e.g.}}
\newcommand{\USM}{{\texttt{USM}}}
\newcommand{\UDRSM}{{\texttt{UDRSM}}}

\newcommand{\characteristic}{\mathbf{1}}
\newcommand{\Update}{{\text{\texttt{Update}}}}
\newcommand{\DiscreteUpdate}{{\text{\texttt{Discrete-Update}}}\xspace}
\newcommand{\PreProcess}{{\text{\texttt{Pre-Process}}}}
\newcommand{\DiscretePreProcess}{{\text{\texttt{Discrete-PreProcess}}}}
\newcommand{\RSet}{{\mathsf{R}}}
\newcommand{\opt}{z^*}

\newcommand{\EventUpdate}[1][]{{\cE^{\ifx&#1& \else #1 \fi}
_{\text{\textnormal{U}}}}}
\newcommand{\EventGreedy}{{\cE_{\text{\textnormal{G}}}}}
\newcommand{\EventPreProcess}{{\cE_{\text{\textnormal{P}}}}}
\newcommand{\EventGreedyComplement}{{\cE^c_{\text{\textnormal{G}}}}}
\newcommand{\EventPreProcessComplement}{{\cE^c_{\text{\textnormal{P}}}}}

\SetKwIF{WithProbability}{wElseif}{wElse}{with probability}{do}{elif}{else}{end}

\crefname{observation}{Observation}{Observations}
\crefname{assumption}{Assumption}{Assumptions}

\title{\textbf{Unconstrained Submodular Maximization\\with Constant Adaptive Complexity}}
\author{
	Lin Chen\thanks{Yale Institute for Network Science, Yale University.  E-mail: \texttt{lin.chen@yale.edu}.}
	\and
	Moran Feldman\thanks{Depart. of Mathematics and Computer Science, The Open University of Israel. E-mail: \texttt{moranfe@openu.ac.il}.}
	\and
	Amin Karbasi\thanks{Yale Institute for Network Science, Yale University. E-mail: \texttt{amin.karbasi@yale.edu}.}
}

\begin{document}

\maketitle
\input{Abstract}
\pagenumbering{Alph}
\thispagestyle{empty}
\clearpage
\pagenumbering{arabic}
\setcounter{page}{1}

\input{Introduction}
\input{Preliminaries}
\input{AlgorithmMultilinearOracle}

\bibliographystyle{plain}
\bibliography{references}
\appendix
\input{AlgorithmDiscreteSet}
\input{AlgorithmDRSubmodular}

\end{document}

%% file: Abstract.tex
\begin{abstract}
	In this paper, we consider the unconstrained submodular maximization problem. We propose the 
	first 
	algorithm for this problem that achieves a tight $(1/2-\eps)$-approximation 
	guarantee using 
	$\tilde{O}(\eps^{-1})$ 
	adaptive rounds and a linear number of function evaluations. No previously known algorithm 
	for this problem achieves an approximation ratio better than $1/3$ using less than $\Omega(n)$ 
	rounds of adaptivity, where $n$ is the size of the ground set. Moreover, our algorithm easily extends to the maximization of a non-negative continuous DR-submodular function subject to a box constraint, and achieves a tight $(1/2-\eps)$-approximation guarantee for this problem while keeping the same adaptive and 
	query complexities.

\medskip
\noindent \textbf{Keywords:} Submodular maximization, low adaptive complexity, parallel computation
\end{abstract}

%% file: Introduction.tex
\section{Introduction}

Faced with the massive data sets ubiquitous in many modern machine learning and data mining applications, there 
has been a tremendous interest in developing  parallel and scalable optimization algorithms. At the 
heart of designing 
such algorithms, there is an inherent trade-off between the number of adaptive sequential rounds of 
parallel computations (also known as
\emph{adaptive complexity}), the total number of objective function evaluations (also known as 
\emph{query complexity}) and the resulting solution quality. 

In the context of submodular maximization, the above trade-off has recently received a growing interest. 
We say that a set function $f\colon 2^\cN \to \bR$ on a finite ground set $\cN$ of size $n$ is
\textit{submodular} if it 
satisfies 
\[f(S\cup\{e\})-f(S)\geq f(T\cup\{e\})-f(T) \quad \text{for every } S\subseteq T\subseteq \cN \text{ 
	and } 
e \in \cN\setminus T \enspace.\]
We also say that such a function is \emph{monotone} if it satisfies $f(S) \leq f(T)$ for every two sets $S 
\subseteq T 
\subseteq \cN$.  
The definition of submodularity intuitively captures \textit{diminishing returns}, which allows submodular functions to faithfully model 
diversity, cooperative costs 
and information gain, making  them increasingly important in various machine learning and 
artificial intelligence
applications \cite{dolhansky2016deep}. Examples include viral marketing \cite{kempe03}, data 
summarization 
\cite{lin2012learning, 
wei2013using}, 
neural network 
interpretation \cite{elenbergDFK17}, active learning \cite{golovin11,guillory2012active}, sensor 
placement 
\cite{krause06nearoptimal}, dictionary 
learning \cite{das2011submodular}, 
compressed sensing 
\cite{elenberg2016restricted} and fMRI 
parcellation \cite{salehi2017submodular}, to name a few. At the same time, submodular functions 
also enjoy tractability as they can be 
minimized exactly and maximizaed approximately in polynomial time.  In 
fact, there has been a surge of novel algorithms to solve submodular maximization problems at 
scale under various models of computation, including centralized \cite{BFNS15,CCPV11,feldman2011unified,nemhauser78},  streaming 
\cite{badanidiyuru2014streaming,buchbinder2015online, chekuri2015streaming, kumar13fast}, 
distributed \cite{barbosa2015power,barbosa2016new,mirrokni2015randomized,mirzasoleiman2013distributed} and 
decentralized \cite{mokhtari2018decentralized} 
frameworks. While  the 
aforementioned works aim to 
obtain tight approximation guarantees, and some other works strove to achieve this goal with a minimal number of functions evaluations 
\cite{badanidiyuru2014fast,buchbinder2014submodular, feldmanHK17, feldman2018less,Mirzasoleiman15}, until recently almost all works on submodular maximization ignored one important  aspect of 
optimization, 
namely, the adaptive 
complexity. More formally,   the adaptive complexity of a submodular maximization procedure is 
the minimum
number of sequential rounds required for implementing it, where in each 
round polynomially-many independent 
function evaluations 
can be executed in parallel \cite{BS18}. All the previously mentioned works may require  $\Omega(n)$ adaptive rounds in the worst case. 
 
A year ago, Balkanski and Singer~\cite{BS18} showed, rather surprisingly, that one can achieve an approximation ratio of 
 $1/3 - \eps$ for maximizing a non-negative monotone submodular function
 subject to a cardinality constraint using $O(\eps^{-1} \log n)$ 
 adaptive rounds. They also proved that no constant factor approximation guarantee can be obtained for this 
 problem in $o(\log n)$ adaptive rounds. The approximation guarantee 
 of~\cite{BS18} was very quickly improved in several independent works~\cite{BRS19,EN19,FMZ19a} to $1 - 
 \nicefrac{1}{e} - \eps$ (using $O(\eps^{-2} \log n)$ adaptive rounds), which almost matches an impossibility 
 result by~\cite{NW78} showing that no polynomial time algorithm can achieve $(1 -\nicefrac{1}{e} +\eps)$-approximation for the problem, regardless of the amount of adaptivity it uses. It should be noted also 
 that~\cite{FMZ19a} manages to achieve the above parameters while keeping the query complexity linear in 
 $n$. An even more recent line of work studies algorithms with low adaptivity for more general 
 submodular maximization problems, which includes problems with non-monotone objective 
 functions and/or constraints beyond the cardinality constraint~\cite{CQ19,ENV18,FMZ18b}. Since 
 all these results achieve constant approximation for problems generalizing the maximization of a 
 monotone submodular function subject to a cardinality constraint, they all inherit the impossibility 
 result of~\cite{BS18}, and thus, use at least $\Omega(\log n)$ adaptive rounds.

In this paper, we study the  \texttt{Unconstrained Submodular Maximization} (\USM) problem 
which asks to find an arbitrary set $S\subseteq\cN$ maximizing  a given non-negative submodular function 
$f(S)$. 
This 
problem was studied by a long list of works~\cite{BF18,BFNS15,FMV11,FNS11,GV11}, culminating 
with a linear time $\nicefrac{1}{2}$-approximation algorithm~\cite{BFNS15}, which was proved to 
be the best possible approximation for the problem by~\cite{FMV11}. Since it does not impose any 
constraints on the solution, {\USM} does not inherit the impossibility result of~\cite{BS18}. In fact, 
it is known  that one 
can get approximation ratios of $\nicefrac{1}{4}$ and $\nicefrac{1}{3}$ for this problem using 
$0$ and $1$ adaptive rounds, respectively~\cite{FMV11}. 
The results of~\cite{FMV11} leave open the question of whether one can get an optimal 
approximation for {\USM} while keeping the number of adaptive rounds independent of $n$. In 
this paper we answer this question in the affirmative. Specifically, we prove the following theorem, 
where the notation $\tilde{O}$ hides a polylogarithmic dependence on $\eps^{-1}$.

\begin{theorem} \label{thm:main_result}
For every constant $\eps > 0$, there is an algorithm that achieves $(\nicefrac{1}{2} - \eps)$-approximation for {\USM} using $\tilde{O}(\eps^{-1})$ adaptive rounds and a query complexity which is linear in $n$.
\end{theorem}

To better understand our result, one should consider the way in which the algorithm is allowed to 
access the objective function $f$. The most natural way to allow such an access is via an oracle 
that given a set $S \subseteq \cN$ returns $f(S)$. Such an oracle is called a \emph{value oracle} 
for $f$. A more powerful way to allow the algorithm access to $f$ is through an oracle known as a 
value oracle for the multilinear extension $F$ of $f$. The multilinear extension of the set function 
$f$ is a function $F\colon [0, 1]^\cN \to \bR$ defined as $F(x) = \bE[f(\RSet(x))]$ for every vector $x \in 
[0, 1]^\cN$, where $\RSet(x)$ is a random set that includes every element $u \in \cN$ with 
probability $x_u$, independently. A value oracle for $F$ is an oracle that given a vector $x \in [0, 
1]^\cN$ returns $F(x)$.

In Section~\ref{sec:algorithm} we describe and analyze an algorithm which satisfies all the requirements 
of Theorem~\ref{thm:main_result} and assumes value oracle access to $F$. Since the multilinear 
extension $F$ can be approximated arbitrarily well using value oracle access to $f$ via sampling~(see, 
\eg, \cite{CCPV11}), it is standard practice to convert algorithms that assume value oracle access to 
$F$ into algorithms that assume such access to $f$. However, a straightforward conversion of this 
kind usually increases the query complexity of the algorithm by a factor of $O(n)$, which is 
unacceptable 
in many applications. Thus, we describe and analyze in Appendix~\ref{sec:discrete} an alternative algorithm which 
satisfies all the requirements of Theorem~\ref{thm:main_result} and assumes value oracle access to 
$f$. 
While this algorithm is not directly related to the algorithm from Section~\ref{sec:algorithm}, the two 
algorithms are based on the same ideas, and thus, we chose to place only the simpler of them in the main part of the paper.

Before concluding this section, we would like to mention that the notion of diminishing returns can be  extended to the 
continuous 
domains as follows. A differentiable function $F\colon\cX\rightarrow\bR^n$, defined over a 
compact set 
$\cX\subset \bR^n$, 
is called 
\textit{DR-submodular} \cite{bian2016guaranteed} if for all vectors $x,y\in \cX$ such that $x\leq y$ we 
have 
$\nabla F(x)\geq \nabla F(y)$---where the inequalities are interpreted coordinate-wise. A canonical example of a 
DR-submodular function is the multilinear extension of a submodular set function.  It has been recently 
shown 
that non-negative DR-submodular functions can be (approximately) maximized over convex bodies 
using first-order methods \cite{bian2016guaranteed,hassanigradient2017, mokhtari2017conditional}. 
Moreover, inspired by 
the double greedy algorithm of~\cite{BFNS15}, it was shown that one can achieve a tight $1/2$-approximation guarantee for the maximization of such functions
subject to 
a box constraint \cite{bian2018optimal,niazadeh2018optimal}. The algorithm we describe in Section~\ref{sec:algorithm} can be easily extended to maximize also arbitrary non-negative DR-submodular functions subject to a box constraint as long as it is possible to evaluate both the objective function and its derivatives. The extended algorithm still achieves a tight $(1/2-\eps)$-approximation guarantee, while keeping its original adaptive and query complexities. The details of the extension are given in Appendix~\ref{sec:dr_submodular}.

\subsection{Our Technique}
All the known algorithms for maximizing a non-negative monotone submodular function subject to a cardinality constraint that use few adaptive rounds update their solutions in iterations. A typical such algorithm decides which elements to add to the solution in a given iteration by considering the set of elements with (roughly) the largest marginal, and then adding as many such elements as possible, as long as the improvement in the value of the solution is roughly linear in the number of added elements. This yields a bound on the number iterations (and thus, adaptive rounds) through the following logic.
\begin{itemize}
	\item The increase stops being roughly linear only when the marginal of a constant fraction of the elements considered decreased significantly. Thus, the set of elements with the maximum marginal decreases in an exponential rate, and after a logarithmic number of iterations no such elements remains, which means that the maximum marginal itself decreases.
\item After the maximum marginal decreases a few times, it becomes small enough that one can argue that there is no need to add additional elements to the solution.
\end{itemize}

A similar idea can be used to decrease the number of adaptive round used by standard algorithms for {\USM} such as the algorithm of~\cite{BFNS15}. However, this results in an algorithm whose adaptive complexity is still poly-logarithmic in $n$. Moreover, both parts of the logic presented above are responsible for this. First, the maximum marginal is only guaranteed to reduce after a logarithmic number of iterations. Second, the maximum marginal has to decrease all the way from $f(OPT)$ to $\eps \cdot f(OPT) / n$, where $OPT$ is an arbitrary optimal solution, which requires a logarithmic number of decreases even when every decrease is by a constant factor.

Getting an adaptive complexity which is independent of $n$ requires us to modify the above framework in two ways. The first modification is that rather than using the maximum marginal to measure the ``advancement'' we have made so far, we use an alternative potential function which is closely related to the gain one can expect from a single element in the next iteration. Since each update adds elements until the gain stops being linear in the number of elements added, we are guaranteed that the gain pair element decreases significantly after every iteration, and so does the potential function.

Unfortunately, the potential function might originally be as large as $2n \cdot f(OPT)$, and the algorithm has to decrease it all the way to at most $\eps \cdot f(OPT)$, which means that the above modification alone cannot make the adaptive complexity independent of $n$. Thus, we also need a second modification which is a pre-processing step designed to decrease the potential to $O(1) \cdot f(OPT)$ in a single iteration. The pre-processing is based on the observation that as long as the gain that can be obtained from a random element is large enough, this gain overwhelms any loss that can be incurred due to this element. Thus, one can evolve the solution in a random way until the potential becomes larger than $f(OPT)$ only by a constant factor.


%% file: Preliminaries.tex
\section{Preliminaries} \label{sec:preliminaries}

Given a set $S \subseteq \cN$, we denote by $\characteristic_S$ the characteristic vector of $S$, \ie, a vector that contains $1$ in the coordinates corresponding to elements of $S$ and $0$ in the remaining coordinates. Additionally, given vectors $x, y \in [0, 1]^\cN$, we denote by $x \vee y$ and $x \wedge y$ their coordinate-wise maximum and minimum, respectively. Similarly, we write $x < y$ and $x \leq y$ when these inequalities hold coordinate-wise.

Given an element $u \in \cN$ and a vector $x \in [0, 1]^\cN$, we denote by $\partial_u F(x)$ the partial derivative of the multilinear extension $F$ with respect to the $u$-coordinate of $x$. One can note that, due to the multilinearity of $F$, $\partial_u F(x)$ obeys the equality
\[
	\partial_u F(x)
	=
	F(x \vee \characteristic_{\{u\}}) - F(x \wedge \characteristic_{\cN \setminus \{u\}})
	\enspace.
\]
One consequence of this equality is that an algorithm with a value oracle access to $F$ also has access to $F$'s derivatives.
As usual, we denote by $\nabla F(x)$ the gradient of $F$ at the point $x$, \ie, $\nabla F(x)$ is a vector whose $u$-coordinate is $\partial_u F(x)$.

The following is a well-known property that we often use in our proofs.
\begin{observation}
Given the multilinear extension $F$ of a submodular function $f\colon 2^\cN \to \bR$ and two vectors $x, y \in [0, 1]^\cN$ obeying $x \leq y$, $\nabla F(x) \geq \nabla F(y)$.
\end{observation}
\begin{proof}
Let $t$ be a uniformly random vector $t \in [0, 1]^\cN$.
For the sake of the proof it is useful to assume that $\RSet(x) = \{u \in \cN \mid x_u \geq t_u\}$ and $\RSet(y) = \{u \in \cN \mid y_u \geq t_u\}$. Notice that this assumption does not change the distributions of $\RSet(x)$ and $\RSet(y)$, and thus, we still have $F(x) = \bE[f(\RSet(x))]$ and $F(y) = \bE[f(\RSet(y))]$. Furthermore, the assumption yields $\RSet(x) \subseteq \RSet(y)$, which implies (by the submodularity of $f$) that for every element $u \in \cN$ we have
\begin{align*}
	\partial_u F(x)
	={} &
	F(x \vee \characteristic_{\{u\}}) - F(x \wedge \characteristic_{\cN \setminus \{u\}})
	=
	\bE[f(\RSet(x) \cup \{u\}) - f(\RSet(x) \setminus \{u\})]\\
	\geq{} &
	\bE[f(\RSet(y) \cup \{u\}) - f(\RSet(y) \setminus \{u\})]
	=
	F(y \vee \characteristic_{\{u\}}) - F(y \wedge \characteristic_{\cN \setminus \{u\}})
	=
	\partial_u F(y)
	\enspace.
	\qedhere
\end{align*}
\end{proof}

%% file: AlgorithmMultilinearOracle.tex
\section{Algorithm} \label{sec:algorithm}

Consider Theorem~\ref{thm:main_result_actual}, and observe that Theorem~\ref{thm:main_result} follows from it when we allow the algorithm value oracle access to the multilinear extension $F$ of the objective function. In this section, we prove Theorem~\ref{thm:main_result_actual} by describing and analyzing an algorithm that obeys all the properties guaranteed by this theorem.
\begin{theorem} \label{thm:main_result_actual}
For every constant $\eps > 0$, there is an algorithm that assumes value oracle access to the multilinear extension $F$ of the objective function and achieves $(\nicefrac{1}{2} - 44\eps)$-approximation for {\USM} using $O(\eps^{-1})$ adaptive rounds and $O(n\eps^{-2}\log \eps^{-1})$ value oracle queries to $F$.
\end{theorem}

Before presenting the promised algorithm, let us quickly recall the main structure of one of the algorithms used by~\cite{BFNS15} to get an optimal $\nicefrac{1}{2}$-approximation for {\USM}. This algorithm maintains two vectors $x, y \in [0, 1]^\cN$ whose original values are $\characteristic_\varnothing$ and $\characteristic_\cN$, respectively. To update these vectors, the algorithm considers the elements of $\cN$ one after the other in some arbitrary order. When considering an element $u \in \cN$, the algorithm finds an appropriate value $r_u \in [0, 1]$, increases $x_u$ from $0$ to $r_u$ and decreases $y_u$ from $1$ to $r_u$. One can observe that this update rule guarantees two things. First, that $x \leq y$ throughout the execution of the algorithm, and second, that both $x$ and $y$ become equal to the vector $r$ (\ie, the vector whose $u$-coordinate is $r_u$ for every $u \in \cN$) when the algorithm terminates.

The analysis of the algorithm of~\cite{BFNS15} depends on the particular choice of $r_u$ used. Specifically, Buchbinder et al.~\cite{BFNS15} showed that their choice of $r_u$ guarantees that the change in $F(x) + F(y)$ following every individual update of $x$ and $y$ is at least twice the change in $-F(OPT(x, y))$ following this update, where $OPT(x, y) \triangleq (\characteristic_{OPT} \vee x) \wedge y$. Since $x$ and $y$ start as $\characteristic_\varnothing$ and $\characteristic_\cN$, respectively, and end up both as $r$, this yields
\[
	2F(r) - [f(\varnothing) + f(\cN)]
	\geq
	2[F(OPT(\characteristic_\varnothing, \characteristic_\cN)) - F(OPT(r, r))]
	=
	2[f(OPT) - F(r)]
	\enspace,
\]
which implies the $\nicefrac{1}{2}$-approximation ratio of the algorithm by rearrangement and the non-negativity of $f$.

The algorithm that we present in this section is similar to the algorithm of~\cite{BFNS15} in the sense that it also maintains two vectors $x, y \in [0, 1]^\cN$ and updates them in a way that guarantees two things. First, that the inequality $x \leq y$ holds throughout the execution of the algorithm, and second, that the change in $F(x) + F(y)$ following every individual update of $x$ and $y$ is at least (roughly) twice the change in $-F(OPT(x, y))$ following this update. More formally, the properties of our update procedure, which we term {\Update}, are described by the following proposition. In this proposition, and throughout the section, we assume that $n \geq 3$ and $\eps$ is in the range $(0, \nicefrac{1}{3})$.\footnote{If $n < 3$, then {\USM} can be solved in constant time (and adaptivity) by enumerating all the possible solutions; and if $\eps \geq \nicefrac{1}{3}$, then Theorem~\ref{thm:main_result_actual} is trivial.} 
\newcommand{\stepProp}[1][]{The input for {\Update} consists of two vectors $x, y \in [0, 1]^\cN$ and two scalars $\Delta \in (0, 1]$ and $\gamma \geq 0$. If this input obeys $y - x = \Delta \cdot \characteristic_\cN$ (\ie, every coordinate of the vector $y$ is larger than the corresponding coordinate of $x$ by exactly $\Delta$), then {\Update} outputs two vectors $x', y' \in [0, 1]^\cN$ and a scalar $\Delta' \in [0, 1]$ obeying
\begin{compactenum}[(a)]
	\item $y' - x' = \Delta' \cdot \characteristic_\cN$, \ifx&#1& \else \label{item:diff_delta} \fi
	\item either $\Delta' = 0$ or $\characteristic_\cN [\nabla F(x') - \nabla F(y')] \leq \characteristic_\cN [\nabla F(x) - \nabla F(y)] - \gamma$ and \ifx&#1& \else \label{item:potential_decrease} \fi
	\item $[F(x') + F(y')] - [F(x) + F(y)]
	\geq
	2(1 - 3\eps) \cdot [F(OPT(x, y)) - F(OPT(x', y'))] - \gamma(\Delta - \Delta') - 2\eps^2 \cdot \characteristic_\cN[\nabla F(x) -\nabla F(y)]$. \ifx&#1& \else \label{item:gain_loss} \fi
\end{compactenum}
Moreover, {\Update} requires only a constant number of  adaptive rounds and $O(n\eps^{-1} \log \eps^{-1})$ value oracle queries to $F$.}
\begin{proposition} \label{prop:step}
\stepProp[l]
\end{proposition}

One can observe that in addition to the guarantees discussed above, Proposition~\ref{prop:step} also shows that {\Update} significantly decreases the expression $\characteristic_\cN[\nabla F(x) - \nabla F(y)]$. Intuitively, this decrease represents the ``progress'' made by every execution of {\Update}, and it allows us to bound the number of iterations (and thus, adaptive rounds) used by our algorithm. Nevertheless, to make the number of iterations independent of $n$, we need to start with $x$ and $y$ vectors for which the expression $\characteristic_\cN[\nabla F(x) - \nabla F(y)]$ is already not too large. We use a procedure named {\PreProcess} to find such vectors. The formal properties of this procedure are given by the next proposition.
\newcommand{\preProcessProp}[1][]{The input for {\PreProcess} consists of a single value $\tau \geq 0$. If $\tau \geq f(OPT) / 4$, then {\PreProcess} outputs two vectors $x, y \in [0, 1]^\cN$ and a scalar $\Delta \in [0, 1]$ obeying
\begin{compactenum}[(a)]
	\item $y - x = \Delta \cdot \characteristic_\cN$, \ifx&#1& \else \label{item:diff_delta_pre} \fi
	\item either $\Delta = 0$ or $\characteristic_\cN [\nabla F(x) - \nabla F(y)] \leq 16\tau$ and \ifx&#1& \else \label{item:potential_start} \fi
	\item $F(x) + F(y) \geq 2[f(OPT) - F(OPT(x, y))] - 4\eps \cdot f(OPT)$. \ifx&#1& \else \label{item:gain_loss_pre} \fi
\end{compactenum}
Moreover, {\PreProcess} requires only a constant number of adaptive rounds and $O(n / \eps)$ value oracle queries to $F$.}
\begin{proposition} \label{prop:pre-process}
\preProcessProp[l]
\end{proposition}

We defer the presentation of the procedures {\Update} and {\PreProcess} and their analyses to Sections~\ref{ssc:update} and~\ref{ssc:pre-process}, respectively. However, using these procedures we are now ready to present the algorithm that we use to prove Theorem~\ref{thm:main_result_actual}. This algorithm is given as Algorithm~\ref{alg:multilinear_oracle}.

\begin{algorithm}[ht]
\caption{\texttt{Algorithm for {\USM} with Value Oracle Access to $F$}} \label{alg:multilinear_oracle}
\DontPrintSemicolon
Let $\tau \gets F(\nicefrac{1}{2} \cdot \characteristic_\cN)$ and $\gamma \gets 4\eps\tau$.\\
Let $i \gets 0$ and $(x^0, y^0, \Delta^0) \gets \PreProcess( \tau)$.\\
\While{$\Delta^i > 0$}
{
	Let $(x^{i + 1}, y^{i + 1}, \Delta^{i + 1}) \gets \Update(x^i, y^i, \Delta^i, \gamma)$.\\
	Update $i \gets i + 1$.
}
\Return{$\RSet(x^i)$}.
\end{algorithm}

Let us denote by $\ell$ the number of iterations made by Algorithm~\ref{alg:multilinear_oracle}. We begin the analysis of the algorithm with the following lemma, which proves some basic properties of Algorithm~\ref{alg:multilinear_oracle}.
\begin{lemma} \label{lem:conditions_hold_two_sided}
It always holds that $\tau \in [f(OPT)/4, f(OPT)]$, and for every integer $0 \leq i \leq \ell$ it holds that $x^i, y^i \in [0, 1]^\cN$, $\Delta^i \in [0, 1]$ and $x^i + \Delta^i \cdot \characteristic_\cN = y^i$.
\end{lemma}
\begin{proof}
It was proved by Feige et al.~\cite{FMV11} that $F(\nicefrac{1}{2} \cdot \characteristic_\cN) = \bE[f(\RSet(\nicefrac{1}{2} \cdot \characteristic_\cN))] \geq f(OPT) / 4$. In contrast, since $\RSet(\nicefrac{1}{2} \cdot \characteristic_\cN)$ is always a feasible solution, we get
$F(\nicefrac{1}{2} \cdot \characteristic_\cN) = \bE[f(\RSet(\nicefrac{1}{2} \cdot \characteristic_\cN))] \leq \bE[f(OPT)] = f(OPT)$. This completes the proof of the first part of the lemma.

We prove the rest of the lemma by induction. For $i = 0$ the lemma holds by the guarantee of Proposition~\ref{prop:pre-process}. Assume now that the lemma holds for some $0 \leq i - 1 < \ell$, and let us prove it for $i$. By the induction hypothesis we have $x^{i-1}, y^{i-1} \in [0, 1]^\cN$, $\Delta^{i-1} \in [0, 1]$ and $x^{i-1} + \Delta^{i-1} \cdot \characteristic_{\cN} = y^{i-1}$. Moreover, the fact that the $i-1$ iteration was not the last one implies that $\Delta^{i-1} \neq 0$. Hence, all the conditions of Proposition~\ref{prop:step} on the input for {\Update} hold with respect to the execution of this procedure in the $i$-th iteration of Algorithm~\ref{alg:multilinear_oracle}, and thus, the proposition guarantees $x^i, y^i \in [0, 1]^\cN$, $\Delta^i \in [0, 1]$ and $x^i + \Delta^i \cdot \characteristic_\cN = y^i$, as required.
\end{proof}

The last lemma shows that the input for the procedure {\PreProcess} obeys the conditions of Proposition~\ref{prop:pre-process} and the input for the procedure {\Update} obeys the conditions of Proposition~\ref{prop:step} in all the iterations of Algorithm~\ref{alg:multilinear_oracle}. The following lemma uses these facts to get an upper bound on the number of iterations performed by Algorithm~\ref{alg:multilinear_oracle} and a lower bound on the value of the output of this algorithm. We note that the proofs of this lemma and the corollary that follows it resemble the above discussed analysis of the algorithm of~\cite{BFNS15}.
\begin{lemma} \label{lem:initial_bound_two_sided}
$\ell \leq 5\eps^{-1}$ and $F(x^\ell) + F(y^\ell) \geq 2(1 - 3\eps) \cdot [f(OPT) - F(OPT(x^\ell, y^\ell))] - 168\eps \cdot f(OPT)$.
\end{lemma}
\begin{proof}
Consider the potential function $\Phi(i) = \characteristic_\cN [\nabla F(x^i) - \nabla F(y^i)]$, and let us study the change in this potential as a function of $i$. Since $\Delta^i > 0$ for every $i < \ell$ and the conditions of Proposition~\ref{prop:step} are satisfied in all the iterations of Algorithm~\ref{alg:multilinear_oracle}, this proposition guarantees $\Phi(i) \leq \Phi(i-1) - 4\eps\tau$ for every $1 \leq i \leq \ell - 1$. In other words, the potential function decreases by at least $4\eps\tau$ every time that $i$ increases by $1$, and thus, $\Phi(0) \geq \Phi(\ell - 1) + 4\eps\tau(\ell - 1)$. Next, we would like to bound $\Phi(0)$ and $\Phi(\ell - 1)$. Since the conditions of Proposition~\ref{prop:pre-process} are also satisfied, it guarantees that $\Phi(0) = \characteristic_\cN [\nabla F(x^0) - \nabla F(y^0)] \leq 16\tau$. In contrast, the submodularity of $f$ and the inequality $y^{\ell - 1} = x^{\ell - 1} + \Delta^{\ell - 1} \cdot \characteristic_\cN \geq x^{\ell - 1}$ imply $\Phi(\ell - 1) = \characteristic_\cN [\nabla F(x^{\ell - 1}) - \nabla F(y^{\ell - 1})] \geq \characteristic_\cN [\nabla F(x^{\ell - 1}) - \nabla F(x^{\ell - 1})] = 0$.
Combining all the above observations, we get
\[
	\Phi(0)
	\geq
	\Phi(\ell - 1) + 4\eps\tau(\ell - 1)
	\Rightarrow
	16\tau
	\geq
	4\eps\tau(\ell - 1)
	\Rightarrow
	\ell \leq 1 + 4\eps^{-1}
	\leq
	5\eps^{-1}
	\enspace.
\]

Let us now get to proving the second part of the lemma. By Proposition~\ref{prop:pre-process}, we get
\begin{align} \label{eq:pre-process_inequality}
	F(x^0) + F(y^0)
	\geq{} &
	2[f(OPT) - F(OPT(x^0, y^0))] - 4\eps \cdot f(OPT)\\ \nonumber
	\geq{} &
	2(1 - 3\eps) \cdot [f(OPT) - F(OPT(x^0, y^0))] - 4\eps \cdot f(OPT)
	\enspace,
\end{align}
where the second inequality holds since $F(OPT(x^0, y^0))$ is the expected value of $f$ over some distribution of sets, and thus, is upper bounded by $f(OPT)$. Additionally, Proposition~\ref{prop:step} implies that for every $1 \leq i \leq \ell$ we have
\begin{align*}
	[F(x^i) + F(y^i)] - [F(x^{i-1}) + F&(y^{i-1})] - 2(1 - 3\eps) \cdot [F(OPT(x^{i-1}, y^{i-1})) - F(OPT(x^i, y^i))]\\
	\geq{}&
	 - 4\eps\tau(\Delta^{i - 1} - \Delta^i) - 2\eps^2 \cdot \characteristic_\cN[\nabla F(x^{i-1}) - \nabla F(y^{i-1})]\\
	={} &
	 - 4\eps\tau(\Delta^{i - 1} - \Delta^i) - 2\eps^2 \cdot \Phi(i-1)
	\geq
	- 4\eps\tau(\Delta^{i - 1} - \Delta^i) - 32\eps^2\tau
	\enspace,
\end{align*}
where the second inequality holds since we have already proved that $\Phi(0) \leq 16\tau$ and that $\Phi(i)$ is a decreasing function of $i$ in the range $0 \leq i \leq \ell - 1$.
Adding up the last inequality for every $1 \leq i \leq \ell$ and adding Inequality~\eqref{eq:pre-process_inequality} to the sum, we get
\begin{align*}
	F(x^\ell) + F(y^\ell)
	\geq{} &
	2(1 - 3\eps) \cdot [f(OPT) - F(OPT(x^\ell, y^\ell))] - 4\eps\tau(\Delta^0 - \Delta^\ell) - 32\ell\eps^2\tau - 4\eps \cdot f(OPT)\\
	\geq{} &
	2(1 - 3\eps) \cdot [f(OPT) - F(OPT(x^\ell, y^\ell))] - 4\eps\tau - 32\ell\eps^2\tau - 4\eps \cdot f(OPT)\\
	\geq{} &
	2(1 - 3\eps) \cdot [f(OPT) - F(OPT(x^\ell, y^\ell))] - (8\eps + 32\ell\eps^2) \cdot f(OPT)\\
	\geq{} &
	2(1 - 3\eps) \cdot [f(OPT) - F(OPT(x^\ell, y^\ell))] - 168\eps \cdot f(OPT)
	\enspace,
\end{align*}
where the second inequality holds since $\Delta^0 - \Delta^\ell \leq 1$, the third inequality holds since $\tau \leq f(OPT)$ by Lemma~\ref{lem:conditions_hold_two_sided} and the last inequality holds by plugging in the upper bound we have on $\ell$.
\end{proof}

\begin{corollary}
$
	F(x^\ell)
	\geq
	(\nicefrac{1}{2} - 44\eps) \cdot f(OPT)
$.
Hence, the approximation ratio of Algorithm~\ref{alg:multilinear_oracle} is at least $\nicefrac{1}{2} - 44\eps$.
\end{corollary}
\begin{proof}
We first note that the second part of the corollary follows from the first one since the last line of Algorithm~\ref{alg:multilinear_oracle} returns a random set whose expected value, with respect to $f$, is $F(x^\ell)$. Thus, the rest of the proof is devoted to proving the first part of the corollary.

Observe that $\Delta^\ell = 0$ because otherwise the algorithm would not have stopped after $\ell$ iterations. Thus, $y^\ell = x^\ell + \Delta^\ell \cdot \characteristic_\cN = x^\ell$ and $OPT(x^\ell, y^\ell) = (\characteristic_{OPT} \vee x^\ell) \wedge y^\ell = x^\ell$. Plugging these observations into the guarantee of Lemma~\ref{lem:initial_bound_two_sided}, we get
\[
	2F(x^\ell) \geq 2(1 - 3\eps) \cdot [f(OPT) - F(x^\ell)] - 168\eps \cdot f(OPT)
	\enspace,
\]
and the corollary now follows immediately by rearranging the last inequality and using the non-negativity of $F$.
\end{proof}

To complete the proof of Theorem~\ref{thm:main_result_actual} we still need to upper bound the adaptivity of Algorithm~\ref{alg:multilinear_oracle} and the number of value oracle queries that it uses, which is done by the next lemma.
\begin{lemma}
The adaptivity of Algorithm~\ref{alg:multilinear_oracle} is $O(\eps^{-1})$, and it uses $O(n\eps^{-2} \log \eps^{-1})$ value oracle queries to $F$.
\end{lemma}
\begin{proof}
Except for the value oracle queries used by the procedures {\Update} and {\PreProcess}, Algorithm~\ref{alg:multilinear_oracle} uses only a single value oracle query (for evaluating $\tau$). Thus, the adaptivity of Algorithm~\ref{alg:multilinear_oracle} is at most
\begin{equation} \label{eq:adaptivity_two_sided}
	1 + (\text{adaptivity of {\PreProcess}}) + \ell \cdot (\text{adaptivity of {\Update}})
	\enspace,
\end{equation}
and the number of oracle queries it uses is at most
\begin{equation} \label{eq:queries_two_sided}
	1 + (\text{value oracle queries used by {\PreProcess}})  + \ell \cdot (\text{value oracle queries used by {\Update}})
	\enspace.
\end{equation}
Proposition~\ref{prop:step} guarantees that each execution of the procedure {\Update} requires at most $O(1)$ rounds of adaptivity and $O(n\eps^{-1}\log \eps^{-1})$ oracle queries, and Proposition~\ref{prop:pre-process} guarantees that the single execution of the procedure {\PreProcess} requires at most $O(1)$ rounds of adaptivity and $O(n / \eps)$ oracle queries. Plugging these observations and the upper bound on $\ell$ given by Lemma~\ref{lem:initial_bound_two_sided} into~\eqref{eq:adaptivity_two_sided} and~\eqref{eq:queries_two_sided}, we get that the adaptivity of Algorithm~\ref{alg:multilinear_oracle} is at most
\[
	1 + O(1) + 5\eps^{-1} \cdot O(1)
	=
	O(\eps^{-1})
	\enspace,
\]
and its query complexity is at most
\[
	1 + O(n / \eps) + 5\eps^{-1} \cdot O(n\eps^{-1}\ln \eps^{-1})
	=
	O(n\eps^{-2}\log \eps^{-1})
	\enspace.
	\qedhere
\]
\end{proof}

\subsection{The Procedure {\Update}} \label{ssc:update}

In this section we describe the promised procedure {\Update} and prove that it indeed obeys all the properties guaranteed by Proposition~\ref{prop:step}. Let us begin by recalling Proposition~\ref{prop:step}.
\begin{repproposition}{prop:step}
\stepProp
\end{repproposition}

The procedure {\Update} itself appears as Algorithm~\ref{alg:update} and consists of two main steps. In the first step the algorithm calculates for every element $u \in \cN$ a basic rate $r_u \in [0, 1]$ whose intuitive meaning is that if an update of size $\delta$ is selected during the second step, then $x_u$ will be increased by $\delta r_u$ and $y_u$ will be decreased by $\delta (1 - r_u)$. Thus, we are guarantees that the difference $y_u - x_u$ decreases by $\delta$ for all the elements of $\cN$. We also note that the formula for calculating the basic rate $r_u$ is closely based on the update rule used by the algorithm of~\cite{BFNS15}.

To understand the second step of {\Update}, observe that $F(x') + F(y')$ can be rewritten in terms of the $\delta$ selected as $F(x + \delta r) + F(y - \delta(\characteristic_\cN - r))$, whose derivative according to $\delta$ is
\[
	r \cdot \nabla F(x + \delta r) - (\characteristic_\cN - r) \cdot \nabla F(y - \delta(\characteristic_\cN - r))
	\enspace.
\]
For $\delta = 0$ this derivative is $r \cdot \nabla F(x) - (\characteristic_\cN - r) \cdot \nabla F(y) = ar + b(\characteristic_\cN - r)$, and the algorithm looks for the minimum $\delta \in [0, \Delta]$ for which the derivative becomes significantly smaller than that. For efficiency purposes, the algorithm only checks possible $\delta$ values out of an exponentially increasing series of values rather than every possible $\delta$ value. The algorithm then makes an update of size $\delta$. Since $\delta$ is (roughly) the first $\delta$ value for which the derivative decreased significantly compared to the original derivative, making a step of size $\delta$ using rates calculated based on the marginals at $x$ and $y$ makes sense. Moreover, since the derivative does decrease significantly after a step of size $\delta$, $\characteristic_\cN[\nabla F(x') - \nabla F(y')]$ should intuitively be significantly smaller than $\characteristic_\cN[\nabla F(x) - \nabla F(y)]$, which is one of the guarantees of Proposition~\ref{prop:step}.
\begin{algorithm}[ht]
\caption{$\Update(x, y, \Delta, \gamma)$} \label{alg:update}
\DontPrintSemicolon
Let $a \gets \nabla F(x)$ and $b \gets -\nabla F(y)$.\\
\For{every $u \in \cN$}{
	\lIf{$a_u > 0$ and $b_u > 0$}{$r_u \gets a_u / (a_u + b_u)$.}
	\lElseIf{$a_u > 0$}{$r_u \gets 1$.}
	\lElse{$r_u \gets 0$.}
}
	
\BlankLine
	
Let $\delta$ be the minimum value in $\{\delta \in [0, \Delta) \mid \exists_{j \in \bZ, j \geq 0}\; \delta = \eps^2(1 + \eps)^j \}$ for which
\[
	r \cdot \nabla F(x + \delta r) - (\characteristic_\cN - r) \cdot \nabla F(y - \delta(\characteristic_\cN - r))
	\leq
	a r + b(\characteristic_\cN - r) - \gamma
	\enspace.
\]
If there is no such value, we set $\delta = \Delta$.\label{line:delta_finding}\\
Let $x' \gets x + \delta r$, $y' \gets y - \delta (\characteristic_\cN - r)$ and $\Delta' = \Delta - \delta$.\\
\Return{$(x', y', \Delta')$}.
\end{algorithm}

We begin the analysis of {\Update} with the following observation, which states some useful properties of the vectors produced by {\Update}, and (in particular) implies part~\eqref{item:diff_delta} of Proposition~\ref{prop:step}. In this observation, and in the rest of the section, we implicitly assume that the input of {\Update} obeys all the requirements of Proposition~\ref{prop:step}.
\begin{observation} \label{obs:basic_update}
$\characteristic_\varnothing \leq r \leq \characteristic_\cN$ and $0 \leq \delta \leq \Delta$, and thus, $x \leq x'$, $y' \leq y$ and $\Delta' \in [0, 1]$. Moreover, $y' - x' = \Delta' \cdot \characteristic_\cN$.
\end{observation}
\begin{proof}
To see why $\characteristic_\varnothing \leq r \leq \characteristic_\cN$ holds, consider an arbitrary coordinate $r_u$ of $r$. The only case in which this coordinate is not set to either $0$ or $1$ by {\Update} is when both $a_u$ and $b_u$ are positive, in which case
\[
	r_u
	=
	\frac{a_u}{a_u + b_u}
	\in
	(0, 1)
	\enspace.
\]
We also observe that the definition of $\delta$ implies $0 \leq \delta \leq \Delta$, and thus, we get $x' = x + \delta r \geq x$, $y' = y - \delta (\characteristic_\cN - r) \leq y$ and $\Delta' = \Delta - \delta \in [0, \Delta] \subseteq [0, 1]$.

It remains to prove $y' - x' = \Delta' \cdot \characteristic_\cN$, which follows since
\[
	y' - x'
	=
	(y - \delta(\characteristic_\cN - r)) - (x + \delta r)
	=
	(y - x) - \delta \cdot \characteristic_\cN
	=
	\Delta \cdot \characteristic_\cN - \delta \cdot \characteristic_\cN
	=
	(\Delta - \delta) \cdot \characteristic_\cN
	=
	\Delta' \cdot \characteristic_\cN
	\enspace.
	\qedhere
\]
\end{proof}

Our next objective is to prove part~\eqref{item:potential_decrease} of Proposition~\ref{prop:step}, which shows that {\Update} makes a significant progress when one measures progress in terms of the decrease in the value of the expression $\characteristic_\cN[\nabla F(x) - \nabla F(y)]$.
\begin{lemma} \label{lem:single_iteration_S_properties}
If $\Delta' > 0$, then $\characteristic_\cN [\nabla F(x') - \nabla F(y')] \leq \characteristic_\cN [\nabla F(x) - \nabla F(y)] - \gamma$.
\end{lemma}
\begin{proof}
Note that $\Delta' > 0$ implies $\delta < \Delta$, which only happens when
\[
	r \cdot \nabla F(x + \delta r) - (\characteristic_\cN - r) \cdot \nabla F(y - \delta(\characteristic_\cN - r))
	\leq
	a r + b(\characteristic_\cN - r) - \gamma
	\enspace.
\]
Plugging in the definitions of $a$, $b$, $x'$ and $y'$, the last inequality becomes
\[
	r \cdot \nabla F(x') - (\characteristic_\cN - r) \cdot \nabla F(y')
	\leq
	r \cdot \nabla F(x) - (\characteristic_\cN - r) \cdot \nabla F(y) - \gamma
	\enspace.
\]
To remove the vector $r$ from this inequality, we add to it the two inequalities $(\characteristic_\cN - r) \cdot \nabla F(x') \leq (\characteristic_\cN - r) \cdot \nabla F(x)$ and $-r \cdot \nabla F(y') \leq - r \cdot \nabla F(y)$. Both these inequalities hold due to submodularity since $x \leq x'$ and $y \geq y'$. One can observe that the result of this addition is the inequality guaranteed by the lemma.
\end{proof}

We now get to proving part~\eqref{item:gain_loss} of Proposition~\ref{prop:step}. Towards this goal, we need to find a way to relate the expression $[F(x') + F(y')] - [F(x) + F(y)]$ to the difference $F(OPT(x, y)) - F(OPT(x', y'))$. The next lemma upper bounds the last difference. Let us define $U^+ = \{u \in \cN \mid a_u > 0 \text{ and } b_u > 0\}$.
\newcommand{\lossTwoSided}{$F(OPT(x, y)) - F(OPT(x', y')) \leq \delta \cdot \sum_{u \in U^+} \max\{b_u r_u, a_u(1 - r_u)\}$.}
\begin{lemma} \label{lem:loss_two_sided}
\lossTwoSided
\end{lemma}
\begin{proof}
Using the chain rule, we get
\begin{align*}
	F(OPT(x, y)) - F(OPT(x', y')) 
	&=
	F((\characteristic_{OPT} \vee x) \wedge y) - F((\characteristic_{OPT} \vee x') \wedge y') \\
	={} &
	-\int_0^\delta \frac{d F((\characteristic_{OPT} \vee (x + t r)) \wedge (y - t(\characteristic_\cN - r)))}{dt} dt\\
	={} &
	\int_0^\delta \left\{\sum_{u \in OPT} (1 - r_u) \cdot \partial_u F((\characteristic_{OPT} \vee (x + t r)) \wedge (y - t(\characteristic_\cN - r))) \right.\\&-\left.  \sum_{u \in \cN \setminus OPT} \mspace{-18mu} r_u \cdot \partial_u F((\characteristic_{OPT} \vee (x + t r)) \wedge (y - t(\characteristic_\cN - r))) \right\} dt
	\enspace.
\end{align*}
Using the submodularity of $f$ and the fact that $x \leq (\characteristic_{OPT} \vee (x + t r)) \wedge (y - t(\characteristic_\cN - r)) \leq y$, the rightmost side of the last equation can be upper bounded as follows.
\begin{align*}
	F(OPT(x, y)&) - F(OPT(x', y'))
	\leq
	\int_0^{\delta} \left\{\sum_{u \in OPT} (1 - r_u) \cdot \partial_u F(x) - \sum_{u \in \cN \setminus OPT} \mspace{-18mu} r_u \cdot \partial_u F(y) \right\} dt\\
	={} &
	\int_0^{\delta} \left\{\sum_{u \in OPT} a_u(1 - r_u) + \sum_{u \in \cN \setminus OPT} \mspace{-18mu} b_u r_u \right\} dt
	\leq
	\int_0^{\delta} \sum_{u \in \cN} \max\{b_ur_u, a_u(1 - r_u)\} dt\\
	={} &
	\delta \cdot \sum_{u \in \cN} \max\{b_ur_u, a_u(1 - r_u)\}
	\enspace.
\end{align*}

To complete the proof of the lemma, it remains to observe that for every element $u \in \cN \setminus U^+$ it holds that $\max\{b_u r_u, a_u(1 - r_u)\} = 0$. To see that this is the case, note that every such element $u$ must fall into one out of only two possible options. The first option is that $a_u > 0$ and $b_u \leq 0$, which imply $r_u = 1$, and thus, $\max\{b_u r_u, a_u(1 - r_u)\} = \max\{b_u, 0\} = 0$. The second option is that $a_u \leq 0$, which implies $r_u = 0$, and thus, $\max\{b_u r_u, a_u(1 - r_u)\} = \max\{0, a_u\} = 0$.
\end{proof}

Next, we would like to lower bound $[F(x') + F(y')] - [F(x) + F(y)]$, which we do in Lemma~\ref{lem:gain_two_sided}. However, before we can state and prove this lemma, we need to prove the following technical observation.
\begin{observation} \label{obs:gain_bound}
For every element $u \in \cN$, $a_u r_u + b_u (1 - r_u) \geq \max\{a_u, b_u\} / 2 \geq 0$.
\end{observation}
\begin{proof}
The submodularity of $f$ implies
\[
	a_u + b_u
	=
	\partial_u F(x) - \partial_u F(y)
	\geq
	\partial_u F(x) - \partial_u F(x)
	=
	0
	\enspace,
\]
which yields the second inequality of the observation.

In the proof of the first inequality of the observation we assume for simplicity that $a_u \geq b_u$. The proof for the other case is analogous. There are now three cases to consider. If $b_u$ is positive, then so must be $a_u$ by our assumption, which implies $r_u = a_u / (a_u + b_u)$, and thus,
\[
	a_ur_u + b_u(1 - r_u)
	\geq
	a_ur_u
	=
	\frac{(a_u)^2}{a_u + b_u}
	\geq
	\frac{(a_u)^2}{2a_u}
	=
	\frac{a_u}{2}
	=
	\frac{\max\{a_u, b_u\}}{2}
	\enspace.
\]
The second case we need to consider is when $a_u = 0$. Note that in this case $r_u = 0$, which implies
\[
	a_ur_u + b_u(1 - r_u)
	=
	b_u
	=
	\frac{\max\{a_u, b_u\}}{2}
	\enspace,
\]
where the second equality holds since the inequality $a_u + b_u \geq 0$ proved above and our assumptions that $a_u = 0$ and $a_u \geq b_u$ yield together $a_u = b_u = 0$.
It remains to consider the case in which $b_u \leq 0$ and $a_u \neq 0$. Note that the inequality $a_u + b_u \geq 0$ proved above implies that in this case we have $a_u > 0$, and thus, $r_u = 1$ and
\[
	a_ur_u + b_u(1 - r_u)
	=
	a_u
	\geq
	\frac{\max\{a_u, b_u\}}{2}
	\enspace.
	\qedhere
\]
\end{proof}

Intuitively, the following lemma holds because the way in which $\delta$ is chosen by {\Update} guarantees that there exists a value $\eta$ which is small, either in absolute terms or compared to $\delta$, such that the derivative of $F(x + \delta' r) + F(y - \delta'(\characteristic_\cN - r))$ as a function of $\delta'$ is almost $ar + b(\characteristic_\cN - r)$ for every $\delta' \in [0, \delta - \eta]$.
\begin{lemma} \label{lem:gain_two_sided}
$[F(x') + F(y')] - [F(x) + F(y)] \geq (1 - 3\eps)\delta \cdot \sum_{u \in U^+} [a_u r_u + b_u(1 - r_u)] - \delta\gamma - 2\eps^2 \cdot \characteristic_\cN[\nabla F(x) - \nabla F(y)]$.
\end{lemma}
\begin{proof}
Let us define $\delta(j) = \eps^2 (1 + \eps)^j$ for every $j \geq 0$, and let us denote by $j^*$ the non-negative value for which $\delta = \delta(j^*)$---if such a value does not exist, which can happen only when $\delta = \Delta < \eps^2$, then we set $j^* = 0$. For convenience, we also define $\delta(-1) = 0$, which implies, in particular, that the inequality $\delta(\lceil j^* \rceil - 1) < \delta$ always holds. Using this notation and the chain rule, we get
\begin{align*}
	F(x') - F(x)
	={} &
	\int_0^\delta \frac{dF(x + tr)}{dt} dt
	=
	\int_0^\delta r \cdot \nabla F(x + tr) dt\\
	={} &
	\int_0^{\delta(\lceil j^* \rceil - 1)} r \cdot \nabla F(x + tr) dt
	+
	\int_{\delta(\lceil j^* \rceil - 1)}^{\delta} r \cdot \nabla F(x + tr) dt\\
	\geq{} &
	\int_0^{\delta(\lceil j^* \rceil - 1)} r \cdot \nabla F(x + tr) dt
	+
	\int_{\delta(\lceil j^* \rceil - 1)}^{\delta} r \cdot \nabla F(y) dt
	\enspace,
\end{align*}
where the inequality holds by the submodularity of $f$ since $x + tr \leq x' = y' - \Delta' \cdot \characteristic_\cN \leq y' \leq y$. We note that the range of the second integral on the rightmost side of the last inequality corresponds to the range $[\delta - \eta, \delta]$ from the intuition given before the lemma.

Using an analogous argument we can also get
\[
	F(y') - F(y)
	\geq
	- \int_0^{\delta(\lceil j^* \rceil - 1)} (\characteristic_\cN - r) \cdot \nabla F(y - t(\characteristic_\cN - r)) dt
	- \int_{\delta(\lceil j^* \rceil - 1)}^{\delta} (\characteristic_\cN - r) \cdot \nabla F(x) dt
	\enspace.
\]
Adding this inequality to the previous one yields
\begin{align} \label{eq:increase_bound_two_sided_x}\nonumber
	[F(x') + F(y')] - [F(x) + F(y)]
	\geq{} &
	\int_0^{\delta(\lceil j^* \rceil - 1)} [r \cdot \nabla F(x + tr) - (\characteristic_\cN - r) \cdot \nabla F(y - t(\characteristic_\cN - r))] dt\\
	&+
	\int_{\delta(\lceil j^* \rceil - 1)}^{\delta} [r \cdot \nabla F(y) - (\characteristic_\cN - r) \cdot \nabla F(x)] dt
	\enspace.
\end{align}

Our next objective is to lower bound the second integral on the right hand side of Inequality~\eqref{eq:increase_bound_two_sided_x}. Towards this goal, we need to prove that for every element $u \in \cN \setminus U^+$ we have
\begin{equation} \label{eq:out_of_U}
	r_u \cdot \partial_u F(y) - (1 - r_u) \cdot \partial_u F(x) \geq 0
	\enspace.
\end{equation}
To see that this is the case, note that every element $u \in \cN \setminus U^+$ must fall into one of the following two cases. The first case is $a_u = \partial_u F(x) > 0$ and $b_u = -\partial_u F(y) \leq 0$, which imply $r_u = 1$, and thus, $r_u \cdot \partial_u F(y) - (1 - r_u) \cdot \partial_u F(x) = \partial_u F(y) \geq 0$; and the other case is $a_u = \partial_u F(x) \leq 0$, which imply $r_u = 0$, and thus, $r_u \cdot \partial_u F(y) - (1 - r_u) \cdot \partial_u F(x) = -\partial_u F(x) \geq 0$. Now, that we have proved Inequality~\eqref{eq:out_of_U}, we are ready to lower bound the second integral on the right hand side of Inequality~\eqref{eq:increase_bound_two_sided_x}.
\begin{align*}
	\int_{\delta(\lceil j^* \rceil - 1)}^{\delta} [r \cdot \nabla F(y) - (\characteristic_\cN - r) \cdot \nabla F&(x)] dt
	\geq
	-\int_{\delta(\lceil j^* \rceil - 1)}^{\delta} \sum_{u \in U^+} \max\{ \partial_u F(x), -\partial_u F(y)\} dt\\
	\geq{} &
	-[\delta(\lceil j^*\rceil) - \delta(\lceil j^*\rceil - 1)] \cdot \sum_{u \in U^+} \max\{ \partial_u F(x), -\partial_u F(y)\}
	\enspace.
\end{align*}
where the first inequality follows from Inequality~\eqref{eq:out_of_U} and the second inequality holds since $\delta(\lceil j^*\rceil) \geq \delta$.

Next, we lower bound the first integral on the right hand side of Inequality~\eqref{eq:increase_bound_two_sided_x}. The definitions of $j^*$ and $\delta$ guarantee that for every $0 \leq t < \delta(\lceil j^* \rceil - 1)$ it holds that
\[
	r \cdot \nabla F(x + t r) - (\characteristic_\cN - r) \cdot \nabla F(y - t(\characteristic_\cN - r))
	\geq
	a r + b(\characteristic_\cN - r) - \gamma
	\enspace,
\]
and thus,
\[
	\int_0^{\delta(\lceil j^* \rceil - 1)} [r \cdot \nabla F(x + tr) - (\characteristic_\cN - r) \cdot \nabla F(y - t(\characteristic_\cN - r))] dt
	\geq
	\delta(\lceil j^* \rceil - 1) \cdot [a r + b(\characteristic_\cN - r) - \gamma]
	\enspace.
\]
Plugging all the bounds we have proved back into Inequality~\eqref{eq:increase_bound_two_sided_x} gives us
\begin{align*}
	[F(x') + F(y')] - [F(x) + F(y)]
	\geq{} &
	\delta(\lceil j^* \rceil - 1) \cdot \left[a r + b(\characteristic_\cN - r) - \gamma\right] \\&-[\delta(\lceil j^*\rceil) - \delta(\lceil j^*\rceil - 1)] \cdot \sum_{u \in U^+} \max\{ \partial_u F(x), -\partial_u F(y)\}
	\enspace.
\end{align*}

There are now two cases to consider. If $j^* = 0$, then $\delta(\lceil j^* \rceil - 1) = \delta(-1) = 0$ and $\delta \leq \delta(0) = \eps^2$ (to see why the last inequality holds, recall that $j^* = 0$ can happen only when $\delta = \delta(0)$ or $\delta = \Delta < \delta(0)$), which yields
\begin{align*}
	[F(x') + F(y')] - [&F(x) + F(y)]
	\geq
	- \delta(0) \cdot \sum_{u \in U^+} \max \{ \partial_u F(x), -\partial_u F(y)\}\\
	\geq{} &
	(1 - 3\eps)\delta \cdot \sum_{u \in U^+} \max \{ \partial_u F(x), -\partial_u F(y)\} - 2 \delta(0) \cdot \sum_{u \in U^+} \max\{ \partial_u F(x), -\partial_u F(y)\}\\
	\geq{} &
	(1 - 3\eps)\delta \cdot \sum_{u \in U^+} \max \{ a_u, b_u\} - 2\eps^2 \cdot \sum_{u \in U^+} [\partial_u F(x) -\partial_u F(y)]\\
	\geq{} &
	(1 - 3\eps)\delta \cdot \sum_{u \in U^+} [a_u r_u + b_u(1 - r_u)] - 2\eps^2 \cdot \characteristic_\cN[\nabla F(x) - \nabla F(y)]
	\enspace,
\end{align*}
where the second and last inequalities hold since the submodularity of $f$ and the inequality $y = x + \Delta \cdot \characteristic_\cN \geq x$ imply that for every element $u \in \cN$ we have $\partial_u F(x) - \partial_u F(y) \geq \partial_u F(x) - \partial_u F(x) = 0$; and the third inequality holds due to the definition of $U^+$ and the fact that the sum of two non-negative numbers always upper bounds their maximum.
Consider now the case of $j^* > 0$. In this case $\delta(\lceil j^*\rceil) - \delta(\lceil j^* \rceil - 1) = \eps \cdot \delta(\lceil j^* \rceil - 1)$ and $\delta = \delta(j^*) \geq \delta(\lceil j^* \rceil - 1) \geq (1 + \eps)^{-1} \cdot \delta(j^*) = (1 + \eps)^{-1}\delta$. Thus, 
\begin{align*}
	[F(&x') + F(y')] - [F(x) + F(y)]\\
	\geq{} &
	\delta(\lceil j^* \rceil - 1) \cdot \left[a r + b(\characteristic_\cN - r) - \gamma - \eps \cdot \sum_{u \in U^+} \max\{\partial_u F(x), -\partial_u F(y)\}\right]\\
	\geq{} &
	\delta(\lceil j^* \rceil - 1) \cdot (1 - 2\eps) \cdot \sum_{u \in U^+} [a_u r_u + b_u(1 - r_u)] - \delta(\lceil j^* \rceil - 1) \cdot \gamma\\
	\geq{} &
	(1 + \eps)^{-1} \delta \cdot (1 - 2\eps) \cdot \sum_{u \in U^+} [a_u r_u + b_u(1 - r_u)] - \delta \gamma
	\geq
	(1 - 3\eps)\delta \cdot \sum_{u \in U^+} [a_u r_u + b_u(1 - r_u)] - \delta \gamma
	\enspace,
\end{align*}
where the last three inequalities hold since the inequalities $a_ur_u + b_u(1 - r_u) \geq \max\{a_u, b_u\}/2 = \max\{\partial_u F(x), -\partial_u F(y)\}/2$ and $a_ur_u + b_u(1 - r_u) \geq 0$ hold for every element $u \in \cN$ by Observation~\ref{obs:gain_bound}.
The lemma now follows because in both the case of $j^* > 0$ and the case of $j^* = 0$ we got a guarantee that is at least as strong as the guarantee of the lemma.
\end{proof}

The next corollary completes the proof of part~\eqref{item:gain_loss} of Proposition~\ref{prop:step}. Its proof combines the guarnatees of Lemmata~\ref{lem:loss_two_sided} and~\ref{lem:gain_two_sided}.
\begin{corollary} \label{cor:iteration_conclusion_step}
\begin{align*}
	[F(x') + F(y')] - [F(x) + F(y)]
	\geq{} &
	2(1 - 3\eps) \cdot [F(OPT(x, y)) - F(OPT(x', y'))] - \gamma(\Delta - \Delta') \\&- 2\eps^2 \cdot \characteristic_\cN[\nabla F(x) - \nabla F(y)]
	\enspace.
\end{align*}
\end{corollary}
\begin{proof}
Our first objective is to show that for every element $u \in \cN$ we have
\begin{equation} \label{eq:basic_inequality}
	2 \cdot \max\{b_u r_u, a_u(1 - r_u)\} \leq a_u r_u + b_u(1 - r_u)
	\enspace.
\end{equation}
There are three cases to consider: $r_u = 0$, $r_u = 1$ and $r_u \in (0, 1)$. Let us consider first the case $r_u = 0$. In this case $a_u$ must be non-positive, which reduces Inequality~\eqref{eq:basic_inequality} to $0 \leq b_u$. To see that the last inequality is true, observe that the submodularity of $f$ and that fact that $y = x + \Delta \cdot \characteristic_\cN \geq x$ imply together $b_u = - \partial_u F(y) \geq - \partial_u F(x) = -a_u \geq 0$. Consider next the case that $r_u = 1$. In this case $b_u$ must be non-positive and $a_u$ must be positive, which guarantees that Inequality~\eqref{eq:basic_inequality} holds. It remains to consider the case that $r_u \in (0, 1)$, which implies that $r_u = a_u / (a_u + b_u)$, and thus,
\[
	a_u r_u + b_u(1 - r_u)
	=
	\frac{(a_u)^2 + (b_u)^2}{a_u + b_u}
	\geq
	\frac{2a_ub_u}{a_u + b_u}
	\enspace.
\]
The last inequality implies Inequality~\eqref{eq:basic_inequality}  since $b_u r_u = a_u b_u / (a_u + b_u) = a_u (1 - r_u)$, which completes the proof that Inequality~\eqref{eq:basic_inequality} holds for every $u \in \cN$.

We are now ready to prove the corollary. Observe that
\begin{align*}
	F(OPT(x, y)) - F&(OPT(x', y'))
	\leq
	\delta \cdot \sum_{u \in U^+} \max\{b_ur_u, a_u(1 - r_u)\}
	\leq
	\frac{\delta}{2} \cdot \sum_{u \in U^+} [a_u r_u + b_u(1 - r_u)]\\
	\leq{} &
	\frac{[F(x') + F(y')] - [F(x) + F(y)] + \gamma\delta + 2\eps^2 \cdot \characteristic_\cN[\nabla F(x) - \nabla F(y)] }{2(1 - 3\eps)}\\
	={} &
	\frac{[F(x') + F(y')] - [F(x) + F(y)] + \gamma(\Delta - \Delta') + 2\eps^2 \cdot \characteristic_\cN[\nabla F(x) - \nabla F(y)] }{2(1 - 3\eps)}
	\enspace,
\end{align*}
where the first inequality holds by Lemma~\ref{lem:loss_two_sided}, the second by Inequality~\eqref{eq:basic_inequality}, the third holds by Lemma~\ref{lem:gain_two_sided} and the equality holds by the definition of $\Delta'$. The corollary now follows by rearranging this inequality.
\end{proof}

To complete the proof of Proposition~\ref{prop:step} it remains to upper bound the adaptivity of Algorithm~\ref{alg:update} and the number of value oracle queries used by this algorithm.
\begin{lemma}
Algorithm~\ref{alg:update} has constant adaptivity and uses $O(n\eps^{-1}\log \eps^{-1})$ value oracle queries to $F$.
\end{lemma}
\begin{proof}
Observe that all the value oracle queries used by {\Update} can be made in two parallel steps. One step for computing $a$ and $b$, which requires $4n$ value oracle queries to $F$, and one additional step for determining $\delta$. This already proves that Algorithm~\ref{alg:update} has constant adaptivity. To prove the other part of the lemma, we still need to show that $\delta$ can be determined using $O(n\eps^{-1}\ln (n / \eps))$ value oracle queries.

Algorithm~\ref{alg:update} determines $\delta$ by evaluating the inequality appearing on Line~\ref{line:delta_finding} in it for at most
\[
	\lceil \ln_{1 + \eps} (\eps^{-2}) \rceil
	\leq
	1 + \frac{\ln \eps^{-2}}{\ln(1 + \eps)}
	\leq
	1 + \frac{\ln \eps^{-1}}{\eps / 2}
	=
	1 + 2\eps^{-1} \ln \eps^{-1}
	=
	O(\eps^{-1} \log \eps^{-1})
\]
different possible $\delta$ values, where the second inequality holds since $\ln (1 + z) \geq z/2$ for every $z \in [0, 1]$. Each evaluation requires $4n$ value oracle queries, and thus, all the evaluations together require only $O(n\eps^{-1}\log \eps^{-1})$ queries, as promised.
\end{proof}

\subsection{The Procedure {\PreProcess}} \label{ssc:pre-process}

In this section we describe the promised procedure {\PreProcess} and prove that it indeed obeys all the properties guaranteed by Proposition~\ref{prop:pre-process}. Let us begin by recalling Proposition~\ref{prop:pre-process}.
\begin{repproposition}{prop:pre-process}
\preProcessProp
\end{repproposition}

The procedure {\PreProcess} itself appears as Algorithm~\ref{alg:pre-process}. To understand this procedure, consider the expression $F(t \cdot \characteristic_\cN) + F((1 - t) \cdot \characteristic_\cN)$. The derivative of this expression by $t$ is $\characteristic_\cN \cdot [\nabla F(t \cdot \characteristic_\cN) - \nabla F((1 - t) \cdot \characteristic_\cN)]$. Using this observation, we get that {\PreProcess} sets $x = t \cdot \characteristic_\cN$ and $y = (1 - t) \cdot \characteristic_\cN$ for (roughly) the first $t$ for which the derivative becomes upper bounded by $16\tau$. This has two advantages, first that $\characteristic_\cN \cdot [\nabla F(x) - \nabla F(y)]$ is small, which is one of the things we need to guarantee, and second, that $F(x) + F(y)$ is at least (roughly) $16\tau \cdot t \geq 4t \cdot f(OPT)$, which is much larger than the loss that $OPT(x, y)$ can suffer due to an $x$ whose coordinates are all only $t$ and a $y$ whose coordinates are all as close to $1$ as $1 - t$. 
\begin{algorithm}[ht]
\caption{$\PreProcess(\tau)$} \label{alg:pre-process}
Let $\delta$ be the minimum value in $\{\delta' \in [\eps, 1/2) \mid \exists_{j \in \bZ, j \geq 1}\; \delta' = \eps j \}$ for which
\[
	\characteristic_\cN \cdot [\nabla F(\delta \cdot \characteristic_\cN) - \nabla F((1 - \delta) \cdot \characteristic_\cN)]
	\leq
	16\tau
	\enspace.
\]
If there is no such value, we set $\delta = 1/2$.\label{line:delta_finding_x}\\
Let $x \gets \delta \cdot \characteristic_\cN$, $y \gets (1 - \delta) \cdot \characteristic_\cN$ and  $\Delta = 1 - 2\delta$.\\
\Return{$(x, y, \Delta)$}.
\end{algorithm}

We begin the analysis of Algorithm~\ref{alg:pre-process} with the following two quite straightforward observations which prove that this algorithm obeys parts~\eqref{item:diff_delta_pre} and~\eqref{item:potential_start} of Proposition~\ref{prop:pre-process}.

\begin{observation}\label{obs:diff_delta_pre}
$x$ and $y$ are both vectors in $[0, 1]^\cN$, $\Delta$ is a scalar in $[0, 1]$ and they obey $y - x = \Delta \cdot \characteristic_\cN$.
\end{observation}
\begin{proof}
The way in which Algorithm~\ref{alg:pre-process} choose a value for $\delta$ guarantees that $\delta \in [\eps, 1/2] \subseteq [0, 1]$. Thus, $x = \delta \cdot \characteristic_\cN \in [0, 1]^\cN$, $y = (1 - \delta) \cdot \characteristic_\cN \in [0, 1]^\cN$ and
\[
	\Delta
	= 1 - 2\delta
	\in
	[1 - 2\cdot (1/2), 1 - 2\eps]
	\subseteq
	[0, 1]
	\enspace.
\]

To complete the proof of the observation, note that
\[
	y
	=
	(1 - \delta) \cdot \characteristic_\cN
	=
	(\Delta + \delta) \cdot \characteristic_\cN
	=
	x + \Delta \cdot \characteristic_\cN
	\enspace.
	\qedhere
\]
\end{proof}

\begin{observation}\label{obs:potential_start}
If $\Delta > 0$, then $\characteristic_\cN [\nabla F(x) - \nabla F(y)] \leq 16\tau$.
\end{observation}
\begin{proof}
Note that $\Delta = 0$ whenever $\delta = 1/2$. Thus, we only need to consider the case in which $\delta < 1/2$. We observe that in this case the way in which $\delta$ is picked by Algorithm~\ref{alg:pre-process} guarantees that
\[
	\characteristic_\cN \cdot [\nabla F(x) - \nabla F(y)]
	=
	\characteristic_\cN \cdot [\nabla F(\delta \cdot \characteristic_\cN) - \nabla F((1 - \delta) \cdot \characteristic_\cN)]
	\leq
	16\tau
	\enspace.
	\qedhere
\]
\end{proof}

Our next objective is to prove part~\eqref{item:gain_loss_pre} of Proposition~\ref{prop:pre-process}. The following lemma lower bounds the term $OPT(x, y)$ that appears in this part.
\begin{lemma} \label{lem:loss_pre}
$F(OPT(x, y)) \geq \Delta \cdot f(OPT)$.
\end{lemma}
\begin{proof}
For every $\lambda \in [0, 1]$ and vector $z \in [0, 1]^\cN$ we define $T_\lambda(z) = \{u \in \cN \mid z_u \geq \lambda\}$. Using this notation, the Lov\'{a}sz extension $\hat{f}\colon [0, 1]^\cN \to \bR$ of a set function $f$ is defined as
\[
	\hat{f}(z)
	=
	\int_0^1 f(T_\lambda(z)) d\lambda
\]
for every vector $z \in [0, 1]^\cN$.

It is well known that the Lov\'{a}sz extension of a submodular function lower bounds the multilinear extension of the same function (see, \eg, \cite{V13}), and thus,
\begin{align*}
	F(OPT(x, y))
	\geq{} &
	\hat{f}(OPT(x, y))
	=
	\int_0^1 f(T_\lambda(OPT(x, y)) d\lambda\\
	\geq{} &
	\int_{\delta}^{1 - \delta} f(T_\lambda(OPT(x, y)) d\lambda
	=
	\int_{\delta}^{1 - \delta} f(OPT) d\lambda
	=
	(1 - 2\delta) \cdot f(OPT)
	\enspace,
\end{align*}
where the second inequality follows from the non-negativity of $f$ and the second equality holds since the $u$-coordinate of $OPT(x, y) = (\characteristic_{OPT} \vee (\delta \cdot \characteristic_\cN)) \wedge ((1 - \delta) \cdot \characteristic_\cN)$ is equal to $\delta$ if $u \not \in OPT$ and to $1 - \delta$ if $u \in OPT$.
\end{proof}

Next, we prove a lower bound on the sum $F(x) + F(y)$, which also appears in part~\eqref{item:gain_loss_pre} of Proposition~\ref{prop:pre-process}.
\begin{lemma} \label{lem:gain_pre}
$F(x) + F(y) \geq (\delta - \eps) \cdot 16\tau$.
\end{lemma}
\begin{proof}
Using the chain rule we get
\begin{align*}
	F(x)
	={} &
	F(\delta \cdot \characteristic_\cN)
	=
	f(\varnothing) + \int_0^{\delta} \frac{dF(t \cdot \characteristic_\cN)}{dt} dt
	=
	f(\varnothing) + \int_0^{\delta} \characteristic_\cN \cdot \nabla F(t \cdot \characteristic_\cN) dt\\
	\geq{} &
	\int_0^{\delta} \characteristic_\cN \cdot \nabla F(t \cdot \characteristic_\cN) dt
	=
	\int_0^{\delta - \eps} \characteristic_\cN \cdot \nabla F(t \cdot \characteristic_\cN)dt
	+
	\int_{\delta - \eps}^{\delta} \characteristic_\cN \cdot \nabla F(t \cdot \characteristic_\cN) dt
	\enspace,
\end{align*}
where the inequality follows from the non-negativity of $f$. Using an analogous argument we can also get
\[
	F(y)
	\geq
	-
	\int_0^{\delta - \eps} \characteristic_\cN \cdot \nabla F((1 - t) \cdot \characteristic_\cN)dt
	-
	\int_{\delta - \eps}^{\delta} \characteristic_\cN \cdot \nabla F((1 - t) \cdot \characteristic_\cN) dt
	\enspace.
\]
Adding this inequality to the previous one yields
\begin{align*}
	F(x) + F(y)
	\geq{} &
	\int_0^{\delta - \eps} \characteristic_\cN \cdot [\nabla F(t \cdot \characteristic_\cN) - \nabla F((1 - t) \cdot \characteristic_\cN)] dt
	\\&+
	\int_{\delta - \eps}^{\delta} \characteristic_\cN \cdot [\nabla F(t \cdot \characteristic_\cN) - \nabla F((1 - t) \cdot \characteristic_\cN)] dt\\
	\geq{} &
	(\delta - \eps) \cdot 16\tau
	+
	\int_{\delta - \eps}^{\delta} \characteristic_\cN \cdot [\nabla F(t \cdot \characteristic_\cN) - \nabla F((1 - t) \cdot \characteristic_\cN)] dt
	\geq
	(\delta - \eps) \cdot 16\tau
	\enspace,
\end{align*}
where the second inequality holds since the choice of $\delta$ and the submodularity of $f$ guarantee that $\characteristic_\cN \cdot [\nabla F(t \cdot \characteristic_\cN) - \nabla F((1 - t) \cdot \characteristic_\cN)]$ is at least $16\tau$ for every $0 \leq t < \delta - \eps$, and the last inequality holds since the submodularity of $f$ and the fact that $\delta \in [\eps, 1/2]$ imply together $\characteristic_\cN \cdot [\nabla F(t \cdot \characteristic_\cN) - \nabla F((1 - t) \cdot \characteristic_\cN)] \geq 0$ for every $t \in [0, \delta]$.
\end{proof}

Combining the last two lemmata, we can now get part~\eqref{item:gain_loss_pre} 
of Proposition~\ref{prop:pre-process}.
\begin{corollary}\label{cor:lower_bound_Fx_Fy}
If $\tau \geq f(OPT)/4$, then $F(x) + F(y) \geq 2[f(OPT) - F(OPT(x, y))] - 4\eps \cdot f(OPT)$.
\end{corollary}
\begin{proof}
Observe that
\begin{align*}
	F(x) + F(y)
	\geq{} &
	(\delta- \eps) \cdot 16\tau
	=
	(1 - \Delta - 2\eps) \cdot 8\tau
	\geq
	(1 - \Delta - 2\eps) \cdot 2f(OPT)\\
	={} &
	2[f(OPT) - \Delta \cdot f(OPT)] - 4\eps \cdot f(OPT)\\
	\geq{} &
	2[f(OPT) - F(OPT(x, y))] - 4\eps \cdot f(OPT)
	\enspace,
\end{align*}
where the first inequality follows from Lemma~\ref{lem:gain_pre} and the second inequality follows from our assumption that $\tau \geq f(OPT) / 4$ and the observation that $\Delta = 1 - 2\delta \leq 1 - 2\eps$. Additionally, the last inequality follows from Lemma~\ref{lem:loss_pre}.
\end{proof}

To complete the proof of Proposition~\ref{prop:pre-process} it remains to upper bound the adaptivity of Algorithm~\ref{alg:pre-process} and the number of value oracle queries used by it.
\begin{lemma}\label{lem:pre-process_adaptivity}
Algorithm~\ref{alg:pre-process} has constant adaptivity and uses $O(n / \eps)$ value oracle queries to $F$.
\end{lemma}
\begin{proof}
Observe that all the value oracle queries used by {\PreProcess} can be made in a single parallel step, and thus, Algorithm~\ref{alg:update} has constant adaptivity. To prove the other part of the lemma, we need to show that $\delta$ can be determined using $O(n / \eps)$ value oracle queries.

Algorithm~\ref{alg:pre-process} determines $\delta$ by evaluating the inequality appearing on Line~\ref{line:delta_finding_x} in it for
\[
	\left \lfloor \frac{1/2}{\eps} \right\rfloor
	\leq
	\frac{\eps^{-1}}{2}
	\leq
	\eps^{-1}
\]
different possible $\delta$ values. Each evaluation requires $4n$ value 
oracle queries, and thus, all the evaluations together require only $O(n / 
\eps)$ queries, as promised.
\end{proof}

%% file: AlgorithmDiscreteSet.tex
\section{Algorithm for Discrete Set} \label{sec:discrete}

As discussed above, our algorithm from Section~\ref{sec:algorithm} can be implemented using a value oracle access to $f$ (rather than to its multilinear extension $F$) at the cost of increasing the query complexity by a factor of $n$ (\ie, making it roughly quadratic in $n$). In many practical cases, however, such a query complexity is infeasible. Thus, in this section, we present an alternative algorithm that uses only a nearly linear in $n$ number of oracle queries to $f$. The result we prove in this section is formally summarized as 
\cref{thm:main_result_discrete}. One can observe that this theorem implies Theorem~\ref{thm:main_result}.

\begin{theorem} \label{thm:main_result_discrete}
	For every constant $0<\eps \le \nicefrac{1}{208}$, there is an algorithm 
	that achieves $(\nicefrac{1}{2} - 104\eps)$-approximation for {\USM} using 
	$O(\eps^{-1}\ln \eps^{-1})$ adaptive rounds and 
	$O(n\eps^{-4}\ln^3(\eps^{-1}))$ value 
	oracle queries to $f$.
\end{theorem}

The algorithm that achieves the guarantees stated as 
\cref{thm:main_result_discrete} is presented in \cref{alg:discrete_set}.
Recall that \cref{alg:multilinear_oracle} 
maintains two sequences of fractional vectors $ x^i, y^i\in [0,1]^\cN $ that 
preserve the invariant that $x^i \le y^i $. Similarly, \cref{alg:discrete_set} maintains two sequences of discrete sets $ 
X^i, Y^i \subseteq \cN $ that preserves the analogous invariant that $ 
X^i \subseteq Y^i $. In Algorithm~\ref{alg:discrete_set} and its analysis we use a few useful shorthands. Specifically, given a set $S \subseteq \cN$ and an element $u \in \cN$, we use $S + u$ and $S - u$ as shorthands for the union $S \cup \{u\}$ and the expression $S \setminus \{u\}$, respectively. We also denote by $f(u \mid S) \triangleq f(S + u) - f(S)$ the marginal gain of $u$ with respect to $S$.

\begin{algorithm}[ht]
	\caption{\texttt{Algorithm for {\USM} with Value Oracle Access to $f$}} \label{alg:discrete_set}
	\DontPrintSemicolon
	Let $ \tau $ be an estimate of the expectation of  $ 
	f(\RSet(\nicefrac{1}{2}\cdot 
	\characteristic_\cN)) $ obtained by averaging 
	$ m = \lceil 200\ln(6/\eps)\rceil $ 
	independent samples 
	from the distribution of this random expression.\label{ln:tau}\\
	Let $\ell \gets \lceil \eps^{-1} \ln(\eps^{-1}) \rceil$ and $(X^0, Y^0) \gets \DiscretePreProcess( 
	\tau)$.\\
	\lFor{$i$ = $1$ \KwTo $\ell$
		\label{ln:two_sides_condition}}
	{
		Let $(X^{i + 1}, Y^{i + 1}) \gets \DiscreteUpdate(X^i, Y^i)$.
	}
	\Return{$ Z\gets X^\ell \cup \{ u\in Y^\ell\setminus X^\ell\mid 
		f(u\mid X^\ell)  > 0 \} $}.
\end{algorithm}

\cref{alg:discrete_set} invokes two subroutines named {\DiscretePreProcess} and {\DiscreteUpdate}, which are counterparts of the subroutines {\PreProcess} and {\Update} used by Algorithm~\ref{alg:multilinear_oracle}. The subroutine {\DiscretePreProcess} outputs the two initial sets $ X^0 $ and $Y^0$, while {\DiscreteUpdate} updates the two sets at each 
iteration. Algorithm~\ref{alg:discrete_set} and its subroutines are all quite similar to their counterparts from Section~\ref{sec:algorithm}, with the main difference being that Algorithm~\ref{alg:discrete_set} and its subroutines have to approximate via sampling quantities (such as $\tau$) that can be directly determined given oracle access to $F$.

Before getting to the analysis of Algorithm~\ref{alg:discrete_set}, we present the following two 
propositions, which summarize the guarantees of the two subroutines. The proofs 
of these propositions are deferred to 
Sections~\ref{sub:discrete_update} and~\ref{sub:discrete_preprocess}. Throughout this section, 
 we use $ OPT(X, Y)\triangleq (OPT\cup X)\cap Y $, where $X$ and $Y$ are two subsets of $\cN$. We note that this is an overload of notation since we previously defined $OPT(x, y)$ for vectors $x,y\in [0, 1]^\cN$ in a related, but slightly different, way. We also assume throughout the section, like in Section~\ref{sec:algorithm}, that $n \geq 3$, and given an event $\cE$ we denote by $\characteristic[\cE]$ its indicator.

\newcommand{\propDiscreteStep}[1][]{%
	The input for {\DiscreteUpdate} consists of two subsets $X, Y \subseteq 
	\cN$. If $X \subseteq Y$, then there exists an event $\EventUpdate$ of 
	probability at least $1 - \eps^2 \ln^{-1}(\eps^{-1})/8$ such that 
	conditioned on this event {\DiscreteUpdate} outputs two subsets $X'$ and $ 
	Y' $ obeying
	\begin{compactenum}[(a)]
		\item $X\subseteq X'\subseteq Y'\subseteq 
		Y$, \ifx&#1& \else \label{item:discrete_subset_chain} \fi
		\item $\bE\left\{\sum_{u \in Y'\setminus X'}  [f(u \mid X') - f(u \mid 
		Y' - u)]\right\}
		\le (1-\eps) \cdot \sum_{u\in Y\setminus X } [f(u \mid X) - f(u \mid Y 
		- u)]$, and
		\item $\bE[f(X')+f(Y')] - [f(X)+f(Y)]
		\geq{} 
		2(1 - 15\eps) \cdot [f(OPT(X, Y)) -\bE [f(OPT(X', Y'))]] - 
		2\eps^2\ln^{-1}(\eps^{-1}) \cdot \sum_{u \in Y \setminus X} [f(u \mid 
		X) - f(u \mid Y - u)]$. 
	\end{compactenum}
	Moreover, {\DiscreteUpdate} requires a constant number of  adaptivity 
	rounds 
	and $O(n\eps^{-3}\ln^2 (\eps^{-1}))$ value 
	oracle queries to $f$.}
\begin{proposition}\label{prop:discrete_step}
	\propDiscreteStep[l]
\end{proposition}

\newcommand{\propDiscretePreProcess}[1][]{%
The input for {\DiscretePreProcess} consists of a single value $\tau \geq 
	0$. If $\tau \geq \nicefrac{1}{5} \cdot f(OPT) $,
	then there exists an event  $ \EventPreProcess $ of probability at least $1-\eps/3 $ such that conditioning on this event	 {\DiscretePreProcess} outputs two sets $ X $ and $ Y $ obeying
	\begin{compactenum}[(a)]
		\item $ 
		X\subseteq 
		Y $, \ifx&#1& \else 
		\label{item:discrete_initial_subset} \fi
		\item $\bE\left\{\sum_{u \in Y \setminus X} [f(u \mid X) - f(u \mid Y - u)]\right\} \leq 
		60\tau
		$, and \ifx&#1& \else \label{item:discrete_potential_start} \fi
		\item $\bE[f(X)+f(Y)] \geq 2[f(OPT) 
		- 
		\bE[f(OPT(X,Y))] ] - 
		4\eps \cdot f(OPT)$. 
		\ifx&#1& \else \label{item:discrete_gain_loss_pre} \fi
	\end{compactenum}
	Moreover, {\DiscretePreProcess} requires only $O(n\eps^{-3}\ln(\eps^{-1}))$ value oracle queries to $f$ and a constant number of 
	adaptivity 
	rounds.}
\begin{proposition} \label{prop:discrete_pre-process}
\propDiscretePreProcess[l]
\end{proposition}

%

We are now ready to begin the analysis of Algorithm~\ref{alg:discrete_set}. Let $ \EventGreedy $ denote the event that $f(OPT) / 5 \leq \tau \leq f(OPT)$, where $ \tau $ is defined in \cref{ln:tau} of 
\cref{alg:discrete_set}.
\cref{lem:probability_event_greedy} lower bounds the probability of $ \EventGreedy $.

\begin{lemma}\label{lem:probability_event_greedy}
	$ \Pr[\EventGreedy]\ge 1 - \eps /3 $.
\end{lemma}
\begin{proof}
Let $ W_1,W_2,\dotsc, W_m $ denote the $ m $ independent random samples that 
	are used by Algorithm~\ref{alg:discrete_set} to estimate $ f(\RSet(\nicefrac{1}{2}\cdot \characteristic_\cN)) $. We observe that they satisfy $ 0 \le f(\RSet(\nicefrac{1}{2}\cdot 
	\characteristic_\cN))\le f(OPT) $ because $f$ is non-negative and $\RSet(\nicefrac{1}{2} \cdot \characteristic_\cN)$ is always a feasible solution. Since $ \tau = 
	m^{-1}\sum_{i=1}^{m} W_i $, we immediately get from this $\tau \leq f(OPT)$. Additionally, by Hoeffding's inequality,
	\begin{align*}
	\Pr\left[\tau < \frac{f(OPT)}{5}\right]
	\leq{} &
	\Pr\left[\tau < \bE[f(\RSet(\nicefrac{1}{2} \cdot \characteristic_\cN))] - \frac{f(OPT)}{20}\right]\\
	\le{} &
	2\exp\left(-\frac{2m^2 \cdot [f(OPT)/20]^2}{m \cdot [ 
		f(OPT)]^2}\right) = 2\exp\left(-\frac{m}{200}\right) \le \frac{\eps}{3} \enspace,
	\end{align*}
	where the first inequality holds since  $\bE[f(\RSet(\nicefrac{1}{2} \cdot \characteristic_\cN))] \geq f(OPT) / 4$ (recall that this is a known result originally proved by Feige et al.~\cite{FMV11}).
\end{proof}

Observe that, in particular, $\EventGreedy$ implies that $\tau$ fulfills the requirement of 
\cref{prop:discrete_pre-process}. The next lemma shows that together with other events it also shows that the invariant $ X^i \subseteq Y^i $ is maintained throughout the execution of Algorithm~\ref{alg:discrete_set}, which is a requirement of 
\cref{prop:discrete_step}. In this lemma, and in the rest of the section, we denote by $\EventUpdate[i]$ the event $\EventUpdate$ corresponding to the execution of the subroutine {\DiscreteUpdate} in the $i$-th iteration of Algorithm~\ref{alg:discrete_set}.

\begin{lemma} \label{lem:conditions_hold_two_sided_discrete}
For every $0 \leq i \leq \ell$,	conditioned on $ \EventGreedy $, $\EventPreProcess$ and $\{\EventUpdate[i] \mid 1 \leq j \leq i\}$, we have
	$X^i\subseteq Y^i\subseteq \cN$.
\end{lemma}
\begin{proof}
We prove the lemma by induction.
Since we assume that $\EventGreedy$ and $\EventPreProcess$ both hold,	part~(\ref{item:discrete_initial_subset}) of \cref{prop:discrete_pre-process} 
	guarantees that $ X^0 \subseteq Y^0 \subseteq \cN $, which is the base of the induction. Next, let us assume that the lemma holds for $i - 1 \geq 0$, and let us prove it for $i$. The induction hypothesis implies that the requirement of Proposition~\ref{prop:discrete_step} is satisfied. Since we also condition on $\EventUpdate[i]$, part~(\ref{item:discrete_subset_chain}) of this proposition gives us $X^{i-1} \subseteq X^i \subseteq Y^i \subseteq Y^{i-1} \subseteq \cN$, which 
	completes the proof.
\end{proof}


Using the last lemma we can now prove the following technical claim, which will come handy in the proofs appearing later in this section. Consider the function
\[
	\Phi(i) = \bE\left[ \sum_{u\in 
	Y^i\setminus 
	X^i } [f(u \mid X^i) - f(u \mid Y^i - u)] ~\middle|~ \EventGreedy, \EventPreProcess, \{\EventUpdate[j] \mid 1 \leq j \leq i\} \right]\enspace.\]

\begin{lemma} \label{lem:discrete_potential}
For every $0 \leq i \leq \ell$, $\Phi(i) \leq 60 \cdot f(OPT)$. Moreover, $\Phi(\ell) \leq 61\eps \cdot f(OPT)$.
\end{lemma}
\begin{proof}
Since Lemma~\ref{lem:conditions_hold_two_sided_discrete} guarantees that the conditions of Proposition~\ref{prop:discrete_step} are satisfied in iteration number $k$ of Algorithm~\ref{alg:discrete_set} given $\EventGreedy$, $\EventPreProcess$ and $\{\EventUpdate[j] \mid 1 \leq j \leq k-1\}$,
  part~(b) of this proposition guarantees that for every $1 \leq k \leq \ell$ we have
	\begin{align*}
		\Phi(k)
		\leq{} &
		(1 - \eps) \cdot \bE\left[ \sum_{u\in 
	Y^{k - 1}\setminus 
	X^{k - 1} } [f(u \mid X^{k - 1}) - f(u \mid Y^{k - 1} - u)] ~\middle|~ \EventGreedy, \EventPreProcess, \{\EventUpdate[j] \mid 1 \leq j \leq k\} \right]\\
	\leq{} &
	\frac{(1 - \eps) \cdot \Phi(k - 1)}{\Pr[\EventUpdate[k] \mid \EventGreedy, \EventPreProcess, \{\EventUpdate[j] \mid 1 \leq j \leq k - 1\}]}
	\enspace,
	\end{align*}
	where the second inequality holds by the law of total expectation since the submodularity of $f$ guarantees that $f(u \mid X^{k - 1}) - f(u \mid Y^{k - 1} - u)] \geq 0$ for every element $u \in Y^{k-1} \setminus X^{k-1}$. Combining this inequality for every $1 \leq k \leq i$ yields
	\[
			\Phi(i) \leq \frac{(1 - \eps)^i \cdot \Phi(0)}{\prod_{k = 1}^i \Pr[\EventUpdate[k] \mid \EventGreedy, \EventPreProcess, \{\EventUpdate[j] \mid 1 \leq j \leq k - 1\}]}
		\leq
		\frac{60(1 - \eps)^i \cdot f(OPT)}{(1 - \eps^2\ln^{-1}(\eps^{-1})/8)^i}
		\leq
		60 \cdot f(OPT)
		\enspace,
	\]
	where the last inequality holds since $\eps < \nicefrac{1}{3}$ by the assumption of Theorem~\ref{thm:main_result_discrete}, and the second inequality holds since given $\EventGreedy$ the condition of \cref{prop:discrete_pre-process} applies, and thus, $\Phi(0) \leq 60 \tau \leq 60 \cdot f(OPT)$---$\tau \leq f(OPT)$ because $\tau$ is the average value of $f$ over $m$ feasible solutions.
	
	For $i = \ell$, we get
	\begin{align*}
		\Phi(\ell) \leq{} &
		\frac{60(1 - \eps)^\ell \cdot f(OPT)}{(1 - \eps^2\ln^{-1}(\eps^{-1})/8)^\ell}
		\leq
		\frac{60(1 - \eps)^{\eps^{-1} \ln(\eps^{-1})} \cdot f(OPT)}{1 - [1 + \eps^{-1} \ln(\eps^{-1})] \cdot \eps^2\ln^{-1}(\eps^{-1})/8}\\
		\leq{} &
		\frac{60e^{-\ln(\eps^{-1})} \cdot f(OPT)}{1 - \eps(\eps \ln^{-1}(\eps^{-1}) + 1)/8}
		\leq
		\frac{60\eps \cdot f(OPT)}{0.993}
		\leq
		61\eps \cdot f(OPT)
		\enspace,
	\end{align*}
	where the penultimate inequality holds since $\eps < 1/20$ by the assumption of Theorem~\ref{thm:main_result_discrete}.
\end{proof}

The first result that we prove using the above technical lemma is a lower bound on the expected value of $f(X^\ell) + f(Y^\ell)$.
\begin{lemma} \label{lem:discrete_sum_X_Y}
Conditioned on $ \EventGreedy $ and $\EventPreProcess$,
	\[
		\bE[f(X^\ell)+f(Y^\ell)] \geq 2(1-15\eps)[ f(OPT) - \bE[f(OPT(X^{\ell},Y^{\ell}) )] ]
		-127\eps \cdot f(OPT)
		\enspace.
	\]
\end{lemma}
\begin{proof}
Recall that the event $\cE_G$ implies that the condition of Proposition~\ref{prop:discrete_pre-process} is
	satisfied, hence, conditioned on $\cE_G$ and $\cE_P$ we have
	\begin{align} \label{eq:discrete_pre-process_inequality}
	\bE[f(X^0) + f(Y^0)]
	\geq{} &
	2[f(OPT) - \bE[f(OPT(X^0, Y^0))]] - 4\eps \cdot f(OPT)\\ \nonumber
	\nonumber
	\geq{} &
	2(1-15\eps) \cdot [f(OPT) - \bE[f(OPT(X^0, Y^0))]] - 
	4\eps 
	\cdot f(OPT)
	\enspace,
	\end{align}
	where the second inequality holds since $F(OPT(X^0, Y^0))$ is the expected 
	value of $f$ over some distribution of sets, and thus, is upper bounded by 
	$f(OPT)$.
	
	We now would like to prove by a backward induction on $i$ that for every $0 \leq i \leq \ell$ we have conditioned on $\EventGreedy$, $\EventPreProcess$ and $\{\EventUpdate[j] \mid 1 \leq j \leq i\}$ the inequality
	\begin{align} \label{eq:event_addition_inequality}
	\bE[f(X^\ell) &{}+ f(Y^\ell) - f(X^{i}) - f(Y^{i})] \geq
	2(1 - 15\eps) \cdot \bE[f(OPT(X^{i}, Y^{i})) - f(OPT(X^\ell, Y^\ell))]\\ \nonumber
	&
	- 122(\ell - i) \cdot \eps^2\ln^{-1}(\eps^{-1}) 
	\cdot
	f(OPT) - 6f(OPT) \cdot \sum_{j = i + 1}^\ell (1 - \Pr[\EventUpdate[j] \mid \{\EventUpdate[k] \mid i < k < j\}])
	\enspace.
	\end{align}
For $i = \ell$ this inequality is trivial since both its sides are equal to $0$. Thus, let us assume that the inequality holds for $1 \leq i + 1 \leq \ell$, and let us prove it for $i$. Since Lemma~\ref{lem:conditions_hold_two_sided_discrete} guarantees that the conditions of Proposition~\ref{prop:discrete_step} are satisfied in iteration number $i + 1$ of Algorithm~\ref{alg:discrete_set} given $\EventGreedy$, $\EventPreProcess$ and $\{\EventUpdate[j] \mid 1 \leq j \leq i\}$, we get---conditioned on these events and $\EventUpdate[i + 1]$---the inequality
\begin{align*}
	\bE[f(X^{i+1}) + f(Y^{i+1}) - f(X^{i}) - f(Y^{i})] &{}- 2(1 - 15\eps) \cdot 
	\bE[f(OPT(X^{i}, Y^{i})) - f(OPT(X^{i + 1}, Y^{i + 1}))]\\
	\geq{}&
	- 2\eps^2\ln^{-1}(\eps^{-1}) \cdot \bE\mleft[ \sum_{u \in Y^{i} \setminus 
	X^{i}} [f(u 
	\mid 
	X^{i}) -f(u \mid Y^{i} - u)]\mright]
	\enspace.
	\end{align*}
	Adding this inequality to the induction hypothesis for $i + 1$, we get that conditioned on the same set of events we also have
	\begin{align} \label{eq:one_more_update_condition}
	\bE[f(X^\ell&) + f(Y^\ell) - f(X^{i}) - f(Y^{i})]
	\geq
	2(1 - 15\eps) \cdot \bE[f(OPT(X^{i}, Y^{i})) - f(OPT(X^\ell, Y^\ell))] \\ \nonumber
	&
	- \left\{122(\ell - i - 1) \cdot f(OPT) + 2\bE\mleft[ \sum_{u \in Y^{i} \setminus 
	X^{i}} [f(u 
	\mid 
	X^{i}) -f(u \mid Y^{i} - u)]\mright]\right\} \cdot \eps^2\ln^{-1}(\eps^{-1}) \\ \nonumber
	&- 6f(OPT) \cdot \sum_{j = i + 2}^\ell (1 - \Pr[\EventUpdate[j] \mid \{\EventUpdate[k] \mid i + 1 < k < j\}])
	\enspace.
	\end{align}
Since the values of $f(X^\ell) + f(Y^\ell) - f(X^{i}) - f(Y^{i})$ and $f(OPT(X^{i}, Y^{i})) - f(OPT(X^\ell, Y^\ell))$ are always in the ranges $[-2f(OPT), 2f(OPT)]$ and $[-f(OPT), f(OPT)]$, respectively, removing the condition on $\EventUpdate[i + 1]$ changes the expectations of these expressions by at most $6f(OPT) \cdot \Pr[\EventUpdate[i + 1] \mid \EventGreedy, \EventPreProcess, \{\EventUpdate[j] \mid 1 \leq j \leq i\}]$. Additionally, since the submodularity of $f$ guarantees that $f(u \mid X^i) - f(u \mid Y^i - u) \geq 0$ whenever $X^i \subseteq Y^i$,
\begin{align*}
	\bE\mspace{100mu}&\mspace{-100mu}\mleft[ \sum_{u \in Y^{i} \setminus	X^{i}} [f(u \mid X^{i}) -f(u \mid Y^{i} - u)] ~\middle|~ \EventGreedy, \EventPreProcess, \{\EventUpdate[j] \mid 1 \leq j \leq i + 1\}\mright]\\
	\leq{} &
	\frac{\bE\mleft[ \sum_{u \in Y^{i} \setminus	X^{i}} [f(u \mid X^{i}) -f(u \mid Y^{i} - u)] ~\middle|~ \EventGreedy, \EventPreProcess, \{\EventUpdate[j] \mid 1 \leq j \leq i\}\mright]}{\Pr[\EventUpdate[i + 1] \mid \EventGreedy, \EventPreProcess, \{\EventUpdate[j] \mid 1 \leq j \leq i\}]}\\
	\leq{} &
	\frac{\Phi(i)}{1 - \eps^2\ln^{-1}(\eps^{-1})/8}
	\leq
	\frac{60 \cdot f(OPT)}{0.987}
	\leq
	61 \cdot f(OPT)
	\enspace,
\end{align*}
where the penultimate inequality holds by Lemma~\ref{lem:discrete_potential} and the fact that $\eps \leq \nicefrac{1}{3}$ by the assumption of Theorem~\ref{thm:main_result_discrete}. Combining all these observations with Inequality~\eqref{eq:one_more_update_condition}, we get that conditioned on the the events $\EventGreedy$, $\EventPreProcess$ and $\{\EventUpdate[j] \mid 1 \leq j \leq i\}$, but not the event $\EventUpdate[i + 1]$, we still have
\begin{align*}
	\bE[f(X^\ell&) + f(Y^\ell) - f(X^{i}) - f(Y^{i})] \geq
	2(1 - 15\eps) \cdot \bE[f(OPT(X^{i}, Y^{i})) - f(OPT(X^\ell, Y^\ell))]\\ \nonumber
	&
	- 122(\ell - i) \cdot \eps^2\ln^{-1}(\eps^{-1}) 
	\cdot
	f(OPT) - 6f(OPT) \cdot \sum_{j = i + 1}^\ell (1 - \Pr[\EventUpdate[j] \mid \{\EventUpdate[k] \mid i < k < j\}]) 
	\enspace,
	\end{align*}
which completes the proof by induction of Inequality~\eqref{eq:event_addition_inequality}.

Plugging now $i = 0$ into Inequality~\eqref{eq:event_addition_inequality} and adding the result to Inequality~\eqref{eq:discrete_pre-process_inequality}, we get that conditioned on $\EventGreedy$ and $\EventPreProcess$
	\begin{equation*}
	\begin{split}
	\bE[f&(X^\ell)+f(Y^\ell)] - 2(1-15\eps)[ f(OPT) - \bE[f(OPT(X^\ell,Y^\ell) )] ]\\
	\geq{} &-\ell \cdot 122\eps^2\ln^{-1}(\eps^{-1}) \cdot f(OPT) - 4\eps \cdot f(OPT) - 6f(OPT) \cdot \sum_{j = 1}^\ell (1 - \Pr[\EventUpdate[j] \mid \{\EventUpdate[k] \mid 0 < k < j\}])\\
		\geq{} &
		-[1 + \eps^{-1} \ln(\eps^{-1})] \cdot [122\eps^2\ln^{-1}(\eps^{-1}) + 6\eps^2\ln^{-1}(\eps^{-1})/8] \cdot f(OPT) - 4\eps \cdot f(OPT) \\
		={} &
		-[\eps \ln^{-1}(\eps^{-1}) + 1] \cdot 122.75\eps \cdot f(OPT) - 4\eps \cdot f(OPT)
		\geq
		-127\eps \cdot f(OPT)
	\enspace,
	\end{split}
	\end{equation*}
	where the second inequality holds since $\Pr[\EventUpdate[j] \mid \EventGreedy, \EventPreProcess, \{\EventUpdate[k] \mid 0 < k < j\}] \geq 1 - \eps^2 \ln^{-1} (\eps^{-1})/8$, and the last inequality holds since $\eps < 1/110$ by the assumption of Theorem~\ref{thm:main_result_discrete}.
\end{proof}

The last lemma involves the the expected values with respect to $f$ of the sets $X^\ell$, $Y^\ell$ and $OPT(X^\ell, Y^\ell)$. Since $X^\ell \subseteq Y^\ell$ (conditioned on the appropriate events), all these sets are subsets of $Y^\ell$ and include $X^\ell$. The following lemma shows that this implies that $\bE[f(Z)]$ upper bounds their expected values with respect to $f$, up to some error.
\begin{lemma} \label{lem:Z_W_bound}
Conditioned on $ \EventGreedy $, $\EventPreProcess$, we have $\bE[f(Z)] \geq \bE[f(W)] - 63\eps \cdot f(OPT)$ for every (random) set $W$ that obeys $X^\ell \subseteq W \subseteq Y^\ell$ given these events and $\{\EventUpdate[i] \mid 1 \leq i \leq \ell\}$.
\end{lemma}
\begin{proof}
We begin by proving that, conditioned on the events $\EventGreedy$, $\EventPreProcess$ and $\{\EventUpdate[j] \mid 1 \leq j \leq \ell\}$, we have
\begin{equation} \label{eq:Z_W}
	\bE[f(Z)] \geq \bE[f(W)] - \Phi(\ell) \geq \bE[f(W)] - 61\eps^{-1} \cdot f(OPT)
	\enspace.
\end{equation}
The second inequality is an immedaite consequence of Lemma~\ref{lem:discrete_potential}, and thus, we concentrate now on proving the first inequality.
Note that since, given the above mentioned events, both $W$ and $Z$ are subsets of $Y^\ell$ and include $X^\ell$, the submodularity of $f$ guarantees
\begin{align*}
	f(Z)
	\geq{} &
	f(W) + \sum_{u \in Z \setminus W} f(u \mid Y^\ell - u) - \sum_{u \in W \setminus Z} f(u \mid X^\ell)\\
	\geq{} &
	f(W) - \sum_{u \in (Z \setminus W) \cup (W \setminus Z)} \mspace{-27mu} [f(u \mid X^\ell) - f(u \mid Y^\ell - u)]\\
	\geq{} &
	f(W)  - \sum_{u \in Y^\ell \setminus X^\ell} [f(u \mid X^\ell) - f(u \mid Y^\ell - u)]
	\enspace,
\end{align*}
where the second inequality holds since the definition of $Z$ guarantees that $-f(u \mid X^\ell) \leq 0$ for every $u \in Z$ and $f(u \mid Y^\ell - u) \leq f(u \mid X^\ell) \leq 0$ for every $u \not \in Z$; and the last inequality holds since the submodularity of $f$ guarantees $f(u \mid X^\ell) - f(u \mid Y^\ell - u) \geq 0$ for every $u \in Y^\ell \setminus X^\ell$. Taking now expectation of this inequality conditioned on the events $\EventGreedy$, $\EventPreProcess$ and $\{\EventUpdate[j] \mid 1 \leq j \leq \ell\}$, we get Inequality~\eqref{eq:Z_W}.

To prove the lemma, it remains to show that $|\bE[f(S) \mid \EventGreedy, \EventPreProcess] - \bE[f(S) \mid \EventGreedy, \EventPreProcess, \{\EventUpdate[i] \mid 1 \leq i \leq \ell\}]| \leq \eps \cdot f(OPT)$ for every (random) set $S$. To see why that follows from the low of total expectation, notice that $f(S)$ always has a value between $0$ and $f(OPT)$ and that
\begin{align*}
	\Pr\left[\bigcap_{i = 1}^\ell \EventUpdate[i] \mid \EventGreedy, \EventPreProcess \right]
	={} &
	\prod_{i = 1}^\ell \Pr\left[\EventUpdate[i] \mid \EventGreedy, \EventPreProcess, \{\EventUpdate[j] \mid 1 \leq j \leq i - 1\} \right]
	\geq
	(1 - \eps^{2}\ln^{-1}(\eps^{-1})/8)^\ell\\
	\geq{} &
	1 - \ell \cdot \eps^{2}\ln^{-1}(\eps^{-1})/8
	=
	1 - [1 + \eps^{-1} \ln(\eps^{-1})] \cdot \eps^{2}\ln^{-1}(\eps^{-1})/8\\
	={} &
	1 - [\eps\ln^{-1}(\eps^{-1}) + 1] \cdot \eps/8
	\geq
	1 - \eps
	\enspace,
\end{align*}
where the last inequality holds since $\eps < 1/20$ by the assumption of Theorem~\ref{thm:main_result_discrete}.
\end{proof}

The following corollary is a consequence of the last two lemmata. 
\begin{corollary} \label{cor:conditioned_Z}
Conditioned on $\EventGreedy$ and $\EventPreProcess$,
\[
	\bE[f(Z)] \ge{} 
	(\nicefrac{1}{2}-103\eps) \cdot f(OPT) \enspace.
	\]
\end{corollary}
\begin{proof}
All the expectations in this proof are implicitly conditioned on $\EventGreedy$ and $\EventPreProcess$. As discussed above, Lemma~\ref{lem:Z_W_bound} shows that $\bE[f(Z)]$ upper bounds $\bE[f(X^\ell)]$, $\bE[f(Y^\ell)]$ and $\bE[f(OPT(X^\ell, Y^\ell))]$ up to an error of $63\eps \cdot f(OPT)$. Plugging this bound into the guarantee of Lemma~\ref{lem:discrete_sum_X_Y},
we get
\begin{align*}
	2 \bE[f(Z)] + 126\eps \cdot f(OPT)
	\geq{} &
	2(1 - 15\eps)[f(OPT) - \bE[f(Z)] - 63\eps \cdot f(OPT)] \\{}&- 127\eps \cdot f(OPT)
	\enspace.
\end{align*}
The corollary now follows by rearranging this inequality.
\end{proof}

The last corollary gives a guarantee on the approximation ratio of Algorithm~\ref{alg:discrete_set} which is conditioned on $\EventGreedy$ and $\EventPreProcess$. However, to prove Theorem~\ref{thm:main_result_discrete}, we need to prove an unconditional lower bound on $\bE[f(Z)]$, which is done by the next lemma.

\begin{lemma}
$\bE[Z] \geq (\nicefrac{1}{2}-104\eps) \cdot f(OPT)$.
\end{lemma}
\begin{proof}
By the law of total expectation,
\begin{align*}
	\bE[Z]
	={} &
	\Pr[\EventGreedy] \cdot \Pr[\EventPreProcess \mid \EventGreedy] \cdot \bE[Z \mid \EventGreedy, \EventPreProcess] + (1 - \Pr[\EventGreedy, \EventPreProcess]) \cdot \bE[Z \mid \EventGreedyComplement \vee \EventPreProcessComplement]\\
	\geq{} &
	\Pr[\EventGreedy] \cdot \Pr[\EventPreProcess \mid \EventGreedy] \cdot \bE[Z \mid \EventGreedy, \EventPreProcess]\\
	\geq{} &
	(1 - \eps/3)^2 \cdot (\nicefrac{1}{2}-103\eps) \cdot f(OPT)
	\geq
	(\nicefrac{1}{2}-104\eps) \cdot f(OPT)
	\enspace,
\end{align*}
where the first inequality follows from the non-negativity of $f$, and the second inequality follows from the guarantees of Proposition~\ref{prop:discrete_pre-process}, Lemma~\ref{lem:probability_event_greedy} and Corollary~\ref{cor:conditioned_Z}.
\end{proof}

To complete the proof of Theorem~\ref{thm:main_result_discrete}, we still need 
to 
upper bound the adaptivity of Algorithm~\ref{alg:discrete_set} and the 
number of value oracle queries that it uses, which is done by the next lemma.
\begin{lemma}
	The adaptivity of Algorithm~\ref{alg:multilinear_oracle} is 
	$O(\eps^{-1}\ln \eps^{-1})$, 
	and it uses $O(n\eps^{-4}\ln^3(\eps^{-1}))$ value oracle queries 
	to $f$.
\end{lemma}
\begin{proof}
	Except for the value oracle queries used by the two procedures 
	{\DiscreteUpdate} and 
	{\DiscretePreProcess}, Algorithm~\ref{alg:discrete_set} uses only $m$
	value oracle query for calculating $\tau$ and up to $n$ value oracle queries for determining the set $Z$. Since these queries can be done in two adaptive rounds, the adaptivity of 
	Algorithm~\ref{alg:discrete_set} is at most
	\begin{equation} \label{eq:adaptivity_two_sided_discrete}
	2 + (\text{adaptivity of {\DiscretePreProcess}}) + \ell \cdot 
	(\text{adaptivity of 
		{\DiscreteUpdate}})
	\enspace,
	\end{equation}
	and the number of oracle queries it uses is at most
	\begin{align} \label{eq:discrete_queries_two_sided}
	m + n +{} \mspace{100mu}&\mspace{-100mu}(\text{value oracle queries used by {\DiscretePreProcess}})  \\\nonumber &+ \ell 
	\cdot 
	(\text{value oracle queries used by {\DiscreteUpdate}})
	\enspace.
	\end{align}
	Proposition~\ref{prop:discrete_step} guarantees that each execution of the 
	procedure 
	{\DiscreteUpdate} requires at most $O(1)$ rounds of adaptivity and 
	$O(n\eps^{-3}\ln^2(\eps^{-1}))$ oracle queries, 
	and 
	Proposition~\ref{prop:discrete_pre-process} guarantees that the single 
	execution of 
	the procedure {\DiscretePreProcess} requires at most $O(1)$ rounds of 
	adaptivity 
	and $O(n\eps^{-3}\ln(\eps^{-1}))$ oracle queries. Plugging these 
	observations and the values of $m$ and $\ell$
	into~\eqref{eq:adaptivity_two_sided_discrete} 
	and~\eqref{eq:discrete_queries_two_sided}, we 
	get that the adaptivity of Algorithm~\ref{alg:discrete_set} is at most
	\[
	2 + O(1) + (\lceil \eps^{-1} \ln \eps^{-1} \rceil) \cdot O(1)
	=
	O(\eps^{-1}\ln \eps^{-1})
	\enspace,
	\]
	and its query complexity is at most
	\[
	\lceil 200\ln(6/\eps)\rceil + n + O(n\eps^{-3}\ln(\eps^{-1})) + (\lceil \eps^{-1} \ln \eps^{-1} \rceil) \cdot 
	O(n\eps^{-3}\ln^2(\eps^{-1}))
	=
	O(n\eps^{-4}\ln^3(\eps^{-1}))
	\enspace.
	\qedhere
	\]
\end{proof}

\subsection{The Procedure {\DiscreteUpdate}}\label{sub:discrete_update}
In this section we describe the procedure {\DiscreteUpdate} used by Algorithm~\ref{alg:discrete_set},  and prove 
that it obeys all the properties guaranteed by 
Proposition~\ref{prop:discrete_step}. The pseudo-code of the procedure appears as Algorithm~\ref{alg:discrete_update}. We remind the reader that {\DiscreteUpdate} is a 
counterpart of the procedure {\Update} from Section~\ref{ssc:update}.

\begin{algorithm}[ht]
	\caption{$\DiscreteUpdate(X, Y)$} \label{alg:discrete_update}
	\DontPrintSemicolon
	\For{every $u \in \cN$}{
		\If{$u \in Y \setminus X$}
		{
			Let $ a_u\gets f(u \mid X) $ and $ b_u \gets - f(u \mid Y-u)$.\\
			\lIf{$a_u > 0$ and $b_u > 0$}{$r_u \gets a_u / (a_u + b_u)$.}
			\lElseIf{$a_u > 0$}{$r_u \gets 1$.}
			\lElse{$r_u \gets 0$.}
		}
		\Else
		{
			Let $a_u \gets 0$, $ b_u \gets 0$ and $r_u \gets 0$.
		}
	}
	
	\BlankLine
	
	Let $ \gamma \gets \eps \cdot \characteristic_{Y \setminus X} (a + b) $.\\

	Let $\Gamma \gets \{\delta \in [0, 1) \mid \exists_{j \in \bZ, j \geq 0}\; \delta = \eps^2\ln^{-1}(\eps^{-1}) \cdot (1 + \eps)^j \}$.\\
	For every $\delta' \in \Gamma$, let $G(\delta')$ be an estimate of the expectation of
	\[
		\sum_{u \in Y \setminus X} [r_u \cdot f(u \mid X \cup \RSet(\delta' r) - u) - (1 - r_u) \cdot f(u \mid Y \setminus \RSet(\delta'(\characteristic_{Y \setminus X} - r)) - u)]
	\]
	obtained by averaging $ m  = \lceil 
	\eps^{-2}\ln(112\eps^{-3}\ln^2(\eps^{-1}))/2 \rceil$ independent samples from 
	the distribution of this random expression.\\
	Let $\delta$ be the minimum value in $\Gamma$ for which $G(\delta) \leq a r + b(\characteristic_{Y \setminus X} - r) - 2\gamma$. If there is no such value, we set $ \delta \gets 1 $.\label{ln:condition_delta}\\
	
	\BlankLine
	
	Let  $ X'\gets X $ and $ Y'\gets Y $.\\
	\For{every $ u\in Y \setminus X $}{
		\WithProbability{$ r_u $\label{ln:probability_ru}}{
			Update $ X'\gets X'+u $	with probability $ \delta  $. \tcp*{\footnotesize \hspace{-1.5mm}Thus, $u$ is added to $X'$ with prob.\ $\delta r_u$.}
			\label{ln:add_to_X'}
		}
		\wElse{
			Update $ Y'\gets Y'-u $ with probability $ \delta  $. \tcp*{\footnotesize \hspace{-1.5mm}$u$ is removed from $Y'$ with prob.\ $\delta (1 - r_u)$.}
			\label{ln:remove_from_Y}
		}
	}
	\Return{$(X', Y')$}.
\end{algorithm}

	%

Let us recall now Proposition~\ref{prop:discrete_step}.  We remind the reader 
that our objective is to prove that Algorithm~\ref{alg:discrete_update} obeys 
all the properties guaranteed by this proposition.
\begin{repproposition}{prop:discrete_step}
\propDiscreteStep
\end{repproposition}

We begin the analysis of Algorithm~\ref{alg:discrete_update} with the following technical lemma.

\begin{lemma}\label{lem:bound_size_gamma}
	$ |\Gamma| \le 1 + 6\eps^{-1}\ln \eps^{-1} = O(\eps^{-1} \ln \eps^{-1}) $.
\end{lemma}
\begin{proof}
	Every value in $\Gamma$ is a value in the range $[0, 1)$ having the form $\eps^2 \ln^{-1} (\eps^{-1}) 
	\cdot (1 + \eps)^i$ for some non-negative integer $i$. The number of values of this form in the above range is
	\begin{align*}
	\lceil \log_{1+\eps} (\eps^{-2}\ln \eps^{-1}) \rceil 
	\le{} & 1 + \frac{\ln(\eps^{-2}\ln \eps^{-1})}{\ln(1+\eps)} 
	\le 1 + \frac{2\ln \eps^{-1}+\ln \ln \eps^{-1}}{\eps/2} 
	\le 1 + \frac{3\ln \eps^{-1}}{\eps/2} \\
	={} & 1 + 6\eps^{-1}\ln \eps^{-1} = O(\eps^{-1}\ln \eps^{-1}) 
	\enspace,
	\end{align*}
	where the second inequality holds since 
	$\ln (1 + z) \geq z/2$ for every $z \in [0, 1]$ and the third holds since $\ln a \leq a$ for every real value $a$.
\end{proof}

At this point we need to define the event $\EventUpdate$ from Proposition~\ref{prop:discrete_step}. Let us define this event as the event that for every $\delta' \in \Gamma$ the estimate $G(\delta')$ is equal to the expectation it should estimate up to an error of $\gamma$. The following lemma shows that $\EventUpdate$ indeed happens with at least the probability guaranteed by Proposition~\ref{prop:discrete_step}. In this lemma, and in the rest of this section, we implicitly assume that the condition $X \subseteq Y$ of Proposition~\ref{prop:discrete_step} holds.
\begin{lemma}
$\Pr[\EventUpdate] \geq 1 - \eps^2 \ln^{-1}(\eps^{-1})/8 $.
\end{lemma}
\begin{proof}
We will prove that, for every given $\delta' \in \Gamma$, with probability at 
least $1 - \eps^{3}\ln^{-2}(\eps^{-1})/56$ the estimate 
$G(\delta')$ is equal to the expectation of 
\begin{equation} \label{eq:random_expression}
	\sum_{u \in Y \setminus X} [r_u \cdot f(u \mid X \cup \RSet(\delta' r) - u) - (1 - r_u) \cdot f(u \mid Y \setminus \RSet(\delta'(\characteristic_{Y \setminus X} - r)) - u)]
\end{equation}
up to an error of $\gamma$. Note that the lemma follows from this via the union 
bound because Algorithm~\ref{alg:discrete_update} calculates $|\Gamma|$ estimates and \cref{lem:bound_size_gamma} guarantees that
\begin{align*}
	[\eps^{3}\ln^{-2}(\eps^{-1})/56] \cdot |\Gamma|
	\leq{} &
	[\eps^{3}\ln^{-2}(\eps^{-1})/56] \cdot (1 + 6\eps^{-1}\ln\eps^{-1})\\
	={} &
	\eps^{3}\ln^{-2}(\eps^{-1})/56 + 6\eps^2\ln^{-1}(\eps^{-1}) / 56\\
	\leq{} &
	\eps^{2}\ln^{-1}(\eps^{-1})/56 + 6\eps^2\ln^{-1}(\eps^{-1}) / 56
	=
	\eps^2\ln^{-1}(\eps^{-1}) / 8
	\enspace.
\end{align*}

Recall now that $G(\delta')$ is calculated by averaging $m$ independent samples of the random expression~\eqref{eq:random_expression}, and let us denote these samples by $W_1, W_2, \dotsc, W_m$. 
We observe that because of the submodularity of $f$ and the fact that $\RSet(\delta'r)$ and $\RSet(\delta'(\characteristic_{Y \setminus X} - r))$ are both subsets of $Y \setminus X$ we have
\begin{align*}
	W_i
	={} &
	\sum_{u \in Y \setminus X} [r_u \cdot f(u \mid X \cup \RSet(\delta' r) - u) - (1 - r_u) \cdot f(u \mid Y \setminus \RSet(\delta'(\characteristic_{Y \setminus X} - r)) - u)]\\
	\geq{} &
	\sum_{u \in Y \setminus X} [r_u \cdot f(u \mid Y - u) - (1 - r_u) \cdot f(u \mid X)]
	=
	- \sum_{u \in Y \setminus X} [r_ub_u + (1 - r_u)a_u]
\end{align*}
and
\begin{align*}
	W_i
	={} &
	\sum_{u \in Y \setminus X} [r_u \cdot f(u \mid X \cup \RSet(\delta' r) - u) - (1 - r_u) \cdot f(u \mid Y \setminus \RSet(\delta'(\characteristic_{Y \setminus X} - r)) - u)]\\
	\leq{} &
	\sum_{u \in Y \setminus X} [r_u \cdot f(u \mid X) - (1 - r_u) \cdot f(u \mid Y - u)]
	=
	\sum_{u \in Y \setminus X} [a_u r_u + b_u (1 - r_u)]
\end{align*}
for every $1 \leq i \leq m$. Hence, the difference between the maximum and minimum values that $W_i$ can take is $\sum_{u \in Y \setminus X} (a_u + b_u)$. If this difference is $0$, then every sample $W_i$ must be equal to its expectation, and we are done. Thus, we may assume in the rest of the proof that $\sum_{u \in Y \setminus X} (a_u + b_u)$ is positive. Using this assumption we get, by Hoeffding's inequality, that the probability that the average $G(\delta') = m^{-1} \cdot \sum_{i = 1}^m W_i$ of the samples deviates from its expectation by more than $\gamma$ is at most
\begin{align*}
	2\exp\left(-\frac{2m\gamma^2}{(\sum_{u \in Y \setminus X} a_u + b_u)^2}\right)
	={} &
	2\exp(-2m\eps^2)\\
	\leq{} &
	2\exp(-\ln(112\eps^{-3}\ln^2(\eps^{-1})))
	\le
	\eps^{3}\ln^{-2}(\eps^{-1})/56
	\enspace.
	\qedhere
\end{align*}
\end{proof}

Our next objective to prove parts (a) and (b) of \cref{prop:discrete_step}, which is done by the next two claims.
\begin{observation}\label{obs:discrete_basic_update}
$ X\subseteq X' \subseteq Y' \subseteq Y $.
\end{observation}
\begin{proof}
	The set $ X' $ is formed by adding elements to $ X $, and therefore, is a 
	superset of $ X $. Similarly, the set $ Y' $ is formed by removing elements 
	from $ Y $, and therefore, is a subset of $ Y $. Next we show that $ X'\subseteq 
	Y' $, which do by proving that every element $ u\in X' $ belongs also to $Y' $.
	
	There are two cases to consider. In the first case, we assume that $ u\in X' \cap X (=X) $, which 
	implies $ u\in Y $ (since $ X\subseteq Y $) and $ u\notin Y\setminus X $. Since only elements of $Y \setminus X$ can be removed form $Y$ when the set $Y'$ is constructed, this implies that $ u $ remains in $ Y' $. The second case we need to consider is when $u \in X' \setminus X$. In this case $u$ was added to $X'$, which implies that $ u\in Y\setminus X \subseteq Y $ because only elements in $ Y\setminus X $ may 
	be added to $ X' $. Since $ u $ can either be added to $ X' 
	$ or removed from $ Y' $ but not both, the fact that it was added to $X'$ implies that it was not removed from $Y'$, and thus, we get again $ u\in Y' $, as promised.
\end{proof}

\begin{lemma} \label{lem:discrete_single_iteration_S_properties}
Conditioned on the event $\EventUpdate$, we have
$\bE\left\{\sum_{u \in Y'\setminus X'}  [f(u \mid X') - f(u \mid Y' - u)]\right\}
		\le (1-\eps) \cdot \sum_{u\in Y\setminus X } [f(u \mid X) - f(u \mid Y - u)]$.
\end{lemma}
\begin{proof}
We will prove that the lemma holds conditioned on any fixed choice of values for the estimates $\{G(\delta') \mid \delta' \in  \Gamma\}$ as long as every estimate deviates from the expectation it should estimate by up to $\gamma$. One can observe that this will imply that the lemma holds also unconditionally by the law of total expectation since the event $\EventUpdate$ happens exactly when the estimates have this property.

The fact that we fixed the estimates $\{G(\delta') \mid \delta' \in \delta' \in \Gamma\}$ implies that $\delta$ is also deterministic, and thus, we only need to handle the randomness used to construct $X'$ and $Y'$. If $\delta = 1$, then we deterministically have $X' = Y'$, which makes the left side of the inequality we need to prove $0$. Hence, the inequality holds in this case since the submodularity of $f$ and the fact that $X \subseteq Y$ guarantee together that $f(u \mid X) - f(u \mid Y - u) \geq 0$ for every $u \in Y \setminus X$.

In the rest of the proof we consider the case of $\delta < 1$. Note that the distributions of the sets $X'$ and $Y'$ are $X \cup \RSet(\delta r)$ and $Y \setminus \RSet(\delta (\characteristic_{Y \setminus X} - r))$, respectively. Thus, we get in this case
\begin{align} \label{eq:basic_G}
	\bE\mspace{80mu}&\mspace{-80mu}\left\{\sum_{u \in Y \setminus X} [r_u \cdot f(u \mid X' - u) - (1 - r_u) \cdot f(u \mid Y' - u)]\right\}\\ \nonumber
	={} &
	\sum_{u \in Y \setminus X} \{r_u \cdot \bE[f(u \mid X' - u)] - (1 - r_u) \cdot \bE[f(u \mid Y' - u)]\}\\ \nonumber
	={} &
	\sum_{u \in Y \setminus X} \{r_u \cdot \bE[f(u \mid X \cup \RSet(\delta r) - u)] - (1 - r_u) \cdot \bE[f(u \mid Y \setminus \RSet(\delta (\characteristic_{Y \setminus X} - r)) - u)]\}\\ \nonumber
	\leq{} &
	G(\delta) + \gamma
	\leq
	ar + b(\characteristic_{Y \setminus X} - r) - \gamma
	=
	\sum_{u \in Y \setminus X} [r_u \cdot f(u \mid X) - (1 - r_u) \cdot f(u \mid Y - u)] - \gamma
	\enspace,
\end{align}
where the first inequality follows from our assumption that $G(\delta)$ is equal to the expectation it estimates up to an error of $\gamma$, and the second inequality follows from the way $\delta$ is chosen.

Observe now that by the submodularity of $f$, for every $u \in Y \setminus X$, we have
\[
	f(u \mid X' - u)
	\leq
	f(u \mid X)
	\Rightarrow
	\bE[(1 - r_u) \cdot f(u \mid X' - u)]
	\leq
	(1 - r_u) \cdot f(u \mid X)
\]
and 
\[
	f(u \mid Y' - u)
	\geq
	f(u \mid Y - u)
	\Rightarrow
	\bE[-r_u \cdot f(u \mid Y' - u)]
	\leq
	-r_u \cdot f(u \mid Y - u)
	\enspace.
\]
Adding these inequalities (for every such $u$) to Inequality~\eqref{eq:basic_G}, we get
\begin{align*}
	\bE\left\{\sum_{u \in Y \setminus X} [f(u \mid X' - u) - f(u \mid Y' - u)]\right\}
	\leq{} &
	\sum_{u \in Y \setminus X} [f(u \mid X) - f(u \mid Y - u)] - \gamma\\
	={} &
	(1 - \eps) \cdot \sum_{u \in Y \setminus X} [f(u \mid X) - f(u \mid Y - u)]
	\enspace.
\end{align*}

To complete the proof of the lemma, it remains to observe that
\[
	\sum_{u \in Y' \setminus X'} [f(u \mid X') - f(u \mid Y' - u)]
	\leq
	\sum_{u \in Y \setminus X} [f(u \mid X' - u) - f(u \mid Y' - u)]
\]
because for every element $u \in Y' \setminus X'$ we have
\[
	f(u \mid X') - f(u \mid Y' - u)]
	=
	f(u \mid X' - u) - f(u \mid Y' - u)]
	\enspace,
\]
and for every other element $u \in Y \setminus X$ the submodularity of $f$ and the fact that $X' \subseteq Y'$ imply together
\[
	f(u \mid X' - u) - f(u \mid Y' - u)]
	\geq
	0
	\enspace.
	\qedhere
\]
\end{proof}

Proving part (c) of Proposition~\ref{prop:discrete_step} is slightly more involved, and the following few claims are devoted to its proof. Let us define $U^+ = \{u \in Y \setminus X \mid a_u > 0 \text{ and } b_u > 0\}$.

\begin{lemma}\label{lem:discrete_loss_two_sided}
Conditioned on the particular value chosen for $\delta$,
\[
	f(OPT(X,Y))-\bE[f(OPT(X',Y'))] \le \delta \cdot
	\sum_{u \in U^+} \max\{b_u r_u, a_u(1 -r_u)\} \enspace.
\]
\end{lemma}
\begin{proof}
Fix an arbitrary order $u_1, u_2, \dotsc, u_{|Y \setminus X|}$ for the elements of $Y \setminus X$, and let us denote by $X^i$, for every $0 \leq i \leq |Y \setminus X|$, a set which agrees with $X'$ on the elements $u_1, \dotsc, u_i$ and with the set $X$ on all other elements (formally, $X^i = X \cup (X' \cap \{u_1, u_2, \dotsc, u_i\})$). Similarly, let $Y^i$ be a set which agrees with $Y'$ on the elements $u_1, \dotsc, u_i$ and with the set $Y$ on all other elements (formally, $Y^i = Y' \cup (Y \setminus \{u_{1}, u_{2}, \dotsc, u_i\})$).

Fix now some $1 \leq i \leq |Y \setminus X|$, and let us bound the expectation of the difference $f(OPT(X^{i-1},\allowbreak Y^{i-1})) - f(OPT(X^i, Y^i))$. There are two cases to consider. The first case is when $u_i \not \in OPT$. In this case $OPT(X^{i-1}, Y^{i-1}) = OPT(X^i, Y^i)$ unless $u_i$ ends up in $X'$, and thus, we get
\begin{align*}
	f(OPT(X^{i-1}, Y^{i-1})) - f(OPT(X^i, Y^i))
	={} &
	\characteristic[u_i \in X'] \cdot [- f(u \mid OPT(X^{i-1}, Y^{i-1}))]\\
	\leq{} &
	\characteristic[u_i \in X'] \cdot [- f(u \mid Y - u)]
	=
	b_{u_i} \cdot \characteristic[u_i \in X']
	\enspace,
\end{align*}
where the inequality follows from the submodularity of $f$. Taking expectation now over both sides of the last inequality, we get
\[
	\bE[f(OPT(X^{i-1}, Y^{i-1})) - f(OPT(X^i, Y^i))]
	\leq
	b_{u_i} \cdot \bE[\characteristic[u_i \in X']]
	=
	b_{u_i} \cdot \delta r_{u_i}
	\enspace.
\]
The second case we need to consider is the case when $u_i \in OPT$. In this case $OPT(X^{i-1}, Y^{i-1}) = OPT(X^i, Y^i)$ unless $u_i$ ends up outside of $Y'$, and thus, we get
\begin{align*}
	f(OPT(X^{i-1}, Y^{i-1})) - f(OPT(X^i, Y^i))
	={} &
	\characteristic[u_i \not \in Y'] \cdot f(u \mid OPT(X^{i-1}, Y^{i-1}) - u)\\
	\leq{} &
	\characteristic[u_i \not \in Y'] \cdot f(u \mid X)
	=
	a_{u_i} \cdot \characteristic[u_i \not \in Y']
	\enspace,
\end{align*}
where the inequality follows from the submodularity of $f$. Taking expectation now over both sides of the last inequality, we get
\[
	\bE[f(OPT(X^{i-1}, Y^{i-1})) - f(OPT(X^i, Y^i))]
	\leq
	a_{u_i} \cdot \bE[\characteristic[u_i \not \in Y']]
	=
	a_{u_i} \cdot \delta (1 - r_{u_i})
	\enspace.
\]

By combining the results that we got for the two cases, we get that it always holds that
\[
	\bE[f(OPT(X^{i-1}, Y^{i-1})) - f(OPT(X^i, Y^i))]
	\leq
	\delta \cdot \max\{b_{u_i} r_{u_i}, a_{u_i}(1 - r_{u_i})\}
	\enspace.
\]
Summing up this inequality over all $1 \leq i \leq |Y \setminus X|$, we get
\begin{align*}
	f(OPT(X, Y)) - \bE[&f(OPT(X', Y'))]
	=
	\sum_{i = 1}^{|Y\setminus X|} \bE[f(OPT(X^{i-1}, Y^{i-1})) - f(OPT(X^i, Y^i))]\\
	\leq{} &
	\sum_{i = 1}^{|Y\setminus X|} [\delta \cdot \max\{b_{u_i} r_{u_i}, a_{u_i}(1 - r_{u_i})\}]
	=
	\delta \cdot \sum_{u \in Y\setminus X} \mspace{-9mu} \max\{b_{u} r_{u}, a_{u}(1 - r_{u})\}
	\enspace.
\end{align*}
To complete the proof of the lemma, it remains to observe that for every element $u \in (Y \setminus X) \setminus U^+$ it holds that $\max\{b_u r_u, a_u(1 - r_u)\} = 0$. To see that this is the case, note that every such element $u$ must fall into one out of only two possible options. The first option is that $a_u > 0$ and $b_u \leq 0$, which imply $r_u = 1$, and thus, $\max\{b_u r_u, a_u(1 - r_u)\} = \max\{b_u, 0\} = 0$. The second option is that $a_u \leq 0$, which implies $r_u = 0$, and thus, $\max\{b_u r_u, a_u(1 - r_u)\} = \max\{0, a_u\} = 0$.
\end{proof}

\begin{observation} \label{obs:discrete_gain_bound}
	For every element $u \in \cN$, $a_u r_u + b_u (1 - r_u) \geq \max\{a_u, 
	b_u\} / 2 \geq 0$.
\end{observation}
\begin{proof}
If $ u\notin Y\setminus X $, $ a_u=b_u=r_u=0 $, and thus, the equality holds 
trivially. Hence, in the sequel, we assume $ u\in Y\setminus X $. Observe now that by the submodularity of $f$
\[
	a_u + b_u
	=
	f(u \mid X) - f(u \mid Y)
	\geq 0
	\enspace,
\]
hence, at least one of the values $a_u$ or $b_u$ must be positive, which implies the second inequality of the observation.

To prove the first inequality of the observation, we need to consider a few cases. If $ a_u \geq b_u >0 $, then we have 
\[
	a_ur_u+b_u(1-r_u)
	=
	\frac{a_u^2+b_u^2}{a_u+b_u}
	\geq
	\frac{a_u^2}{2a_u}
	=
	\frac{a_u}{2}
	=
	\frac{\max\{a_u, b_u\}}{2}
	\enspace.
\]
The case $b_u \geq a_u > 0$ is analogous, and thus, we skip it. Consider now the case of $ a_u>0 \geq b_u $. In this case $r_u = 1$, and thus,
\[
	a_ur_u+b_u(1-r_u) = a_u
	\geq
	\frac{a_u}{2}
	=
	\frac{\max\{a_u, b_u\}}{2}
	\enspace.
\]
The last case we need to consider is the case of $a_u \leq 0$, which implies $r_u = 0$, and also implies $b_u \geq 0$ since we already proved that $a_u + b_u \geq 0$. Thus, we get in this case
\[
	a_ur_u+b_u(1-r_u)
	=
	b_u
	\geq
	\frac{b_u}{2}
	=
	\frac{\max\{a_u, b_u\}}{2}
	\enspace.
	\qedhere
\]
\end{proof}

\begin{lemma} \label{lem:discrete_gain_two_sided}
Conditioned on any fixed choice for the estimates $\{G(\delta') \mid \delta \in \Gamma\}$ such that every one of these estimates is equal to the expectation of the expression it should estimate up to an error of $\gamma$,
\[
	\bE[f(X') + f(Y')] - [f(X) + f(Y)] \geq (1 - 15\eps)\delta\cdot \sum_{u \in U^+} [a_u r_u + 
	b_u(1 - r_u)] - 2\eps^2\ln^{-1}(\eps^{-1}) \cdot \sum_{u \in Y \setminus X} \mspace{-9mu}(a_u + b_u)
	\enspace.
\]
\end{lemma}
\begin{proof}
Let us define $\delta(j) = \eps^2\ln^{-1}(\eps^{-1}) \cdot (1 + \eps)^j$ for every $j \geq 0$, and let us denote by $j^*$ the non-negative value for which $\delta = 
	\delta(j^*)$. Note that such non-negative value always exists since $ 
	\eps^2 \ln^{-1} (\eps^{-1})\le 1 $. For convenience, we also define 
	$\delta(-1) = 0$, making $ \delta(\cdot) $ a strictly increasing function 
	on $ \{-1\}\cup [0,\infty) $, which implies, in particular, that the 
	inequality 
	$\delta(\lceil j^* \rceil - 1) < \delta$ always holds. For ease of reading, let us denote $\hat{\delta} = \delta(\lceil j^* \rceil - 1)$.
	
Let us now denote by $R_1$ a random set that contains every element $u \in Y \setminus X$ with probability $\delta r_u$, independently, and by $R_2$ a random set that contains every element of $R_1$ with probability $\hat{\delta} / \delta$, independently. Note that $R_2$ is always a subset of $R_1$ and the two sets have the same distributions as $\RSet(\hat{\delta}r)$ and $\RSet(\delta r)$, respectively. Using these observations we now get
\begin{align*}
	\bE[f(X')] - f(&X)
	=
	\bE[f(\RSet(\delta r) \mid X)]
	=
	\bE[f(R_1 \mid X)]
	=
	\bE[f(R_2 \mid X)] + \bE[f(R_1 \setminus R_2 \mid X \cup R_2)]\\
	\geq{} &
	\bE\left[ \sum_{u \in Y \setminus X} \characteristic[u \in R_2] \cdot f(u \mid X \cup R_2 - u)\right] + \bE\left[ \sum_{u \in Y \setminus X} \characteristic[u \in R_1 \setminus R_2] \cdot f(u \mid Y - u)\right]\\
	={} &
	\sum_{u \in Y \setminus X} \Pr[u \in R_2] \cdot \bE[f(u \mid X \cup R_2 - u)] - \sum_{u \in Y \setminus X} \Pr[u \in R_1 \setminus R_2] \cdot b_u\\
	={} &
	\hat{\delta} \cdot \sum_{u \in Y \setminus X} r_u \cdot \bE[f(u \mid X \cup R(\hat{\delta}r) - u)] - \delta \cdot \left(1 - \frac{\hat{\delta}}{\delta}\right) \cdot \sum_{u \in Y \setminus X} r_u b_u
	\enspace,
\end{align*}
where the inequality follows from the submodularity of $f$ and the penultimate equality holds since the distribution of $X \cup R_2 - u$ is independent of the membership of $u$ in $R_2$. Using an analogous argument we can also get
\[
	\bE[f(Y')]-f(Y)
	\geq
	-\hat{\delta} \cdot \sum_{u \in Y \setminus X} (1 - r_u) \cdot \bE[f(u \mid Y \setminus R(\hat{\delta} (\characteristic_{Y \setminus X} - r)) - u)] - \delta \cdot \left(1 - \frac{\hat{\delta}}{\delta}\right) \cdot \sum_{u \in Y \setminus X} (1 - r_u) a_u
	\enspace.
\]
Adding this inequality to the previous one yields
\begin{align}\label{eq:discrete_increase_bound_two_sided_x}
	\bE [f(X'&)+f(Y')] - [f(X)+f(Y)]  \\\nonumber
	\ge{} &
	\hat{\delta} \cdot \sum_{u \in Y \setminus X} \left\{r_u \cdot \bE[f(u \mid X \cup R(\hat{\delta}r) - u)] - (1 - r_u) \cdot \bE[f(u \mid Y \setminus R(\hat{\delta} (\characteristic_{Y \setminus X} - r)) - u)] \right\} \\\nonumber
	& - \delta \cdot \left(1 - \frac{\hat{\delta}}{\delta}\right) \cdot \sum_{u \in Y \setminus X} [r_u b_u + (1 - r_u)a_u]\\ \nonumber
	\ge{} &
	\hat{\delta} \cdot \sum_{u \in Y \setminus X} \left\{r_u \cdot \bE[f(u \mid X \cup R(\hat{\delta}r) - u)] - (1 - r_u) \cdot \bE[f(u \mid Y \setminus R(\hat{\delta} (\characteristic_{Y \setminus X} - r)) - u)] \right\} \\\nonumber
	& - \delta \cdot \left(1 - \frac{\hat{\delta}}{\delta}\right) \cdot \sum_{u \in U^+} [r_u b_u + (1 - r_u)a_u]
	\enspace,
\end{align}
where the last inequality holds since every element of $(Y \setminus X) \setminus U^+$ must obey either $a_u > 0$, $b_u \leq 0$ and $r_u = 1$ or $a_u \leq 0$ and $r_u = 0$, and in both cases we clearly have $r_u b_u + (1 - r_u)a_u \leq 0$. 

There are now two cases to consider. If $j^* = 0$, then $\hat{\delta} = \delta(\lceil j^* \rceil - 1) = \delta(-1) = 0$ and $\delta = \delta(j^*) = \delta(0) = \eps^2 \ln^{-1} (\eps^{-1})$, 
	which yields
	\begin{align*}
	\bE[f(X')+f(Y')] - [f(X)&{}+f(Y)] 
	\geq
	- \delta(0) \cdot \sum_{u \in U^+} [r_u b_u + (1 - r_u)a_u]
	\geq
	- \delta(0) \cdot \sum_{u \in U^+} \max \{ a_u,b_u\}\\
	\geq{} &
	\delta \cdot \sum_{u \in U^+} [a_ur_u + b_u(1 - r_u)] - 2
	\delta(0) \cdot \sum_{u \in U^+} (a_u + b_u)\\
	\geq{} &
	(1 - 15\eps)\delta \cdot \sum_{u \in U^+} [a_ur_u + b_u(1 - r_u)] - 2\eps^2\ln^{-1}(\eps^{-1}) \cdot 
	\sum_{u \in Y \setminus X} (a_u + b_u)
	\enspace,
	\end{align*}
where the third inequality holds by the definition of $U^+$ since the sum of any two non-negative values always upper bounds their maximum, and the last inequality holds since for every $u \in Y \setminus X$ we have $a_u r_u + b_u(1 - r_u) \geq 0$ by Observation~\ref{obs:discrete_gain_bound} and $a_u + b_u = f(u \mid X) - f(u \mid Y - u) \geq f(u \mid X) - f(u \mid X) = 0$ by the submodularity of $f$. It remains to prove the lemma for the case of $j^* > 0$. In this case two things happen. First $\hat{\delta} / \delta \geq (1 + \eps)^{-1}$. Second, since $\hat{\delta} \in \Gamma$ and $\hat{\delta} < \delta$, the way $\delta$ was chosen must imply that $G(\hat{\delta}) \geq a r + b(\characteristic_{Y \setminus X} - r) - 2\gamma$. Recall now that $G(\hat{\delta})$ is an estimate for the expectation of 
\[
	\sum_{u \in Y \setminus X} [r_u \cdot f(u \mid X \cup \RSet(\hat{\delta} r) - u) - (1 - r_u) \cdot f(u \mid Y \setminus \RSet(\hat{\delta}(\characteristic_{Y \setminus X} - r)) - u)]
	\enspace,
\]
and by the assumption of the lemma it is correct up to an error of $\gamma$, which implies that the expectation of the last expression is at least $a r + b(\characteristic_{Y \setminus X} - r) - 3\gamma$. Plugging these observations into Inequality~\eqref{eq:discrete_increase_bound_two_sided_x}, we get
\begin{align*}
	\bE [f(X')+f(Y')] - [&f(X)+f(Y)]\\
	\ge{} &
	\hat{\delta}[ar + (b(\characteristic_{Y \setminus X} - r) - 3\gamma] - \delta \cdot \left(1 - \frac{\hat{\delta}}{\delta}\right) \cdot \sum_{u \in U^+} [r_u b_u + (1 - r_u)a_u]\\
	\geq{} &
	\hat{\delta} \cdot \sum_{u \in Y \setminus X} [a_ur_u + b_u(1 - r_u) - 3\eps(a_u + b_u)] - \eps\delta \cdot \sum_{u \in Y \setminus X} \max\{a_u, b_u\}
	\enspace,
\end{align*}
where the second inequality holds since $\max\{a_u, b_u\} \geq 0$ by Observation~\ref{obs:discrete_gain_bound} for every $u \in Y \setminus X$. The same observation also shows that $a_ur_u + b_u(1 - r_u) \geq \max\{a_u, b_u\}/2 \geq \max\{0, (a_u + b_u)/4\}$. Plugging this bound into the previous inequality gives us
\begin{align*}
	\bE [f(X')+f(Y')] - [&f(X)+f(Y)]
	\ge
	(1 - 12\eps)\hat{\delta} \cdot \sum_{u \in Y \setminus X} [a_ur_u + b_u(1 - r_u)] - \eps\delta \cdot \sum_{u \in Y \setminus X} \mspace{-9mu} \max\{a_u, b_u\}\\
	\geq{} &
	(1 - 13\eps)\delta \cdot \sum_{u \in Y \setminus X} [a_ur_u + b_u(1 - r_u)] - \eps\delta \cdot \sum_{u \in Y \setminus X} \mspace{-9mu} \max\{a_u, b_u\}\\
	\geq{} &
	(1 - 15\eps)\delta \cdot \sum_{u \in Y \setminus X} [a_ur_u + b_u(1 - r_u)]
	\geq
	(1 - 15\eps)\delta \cdot \sum_{u \in U^+} [a_ur_u + b_u(1 - r_u)]
	\enspace.
	\qedhere
\end{align*}

\end{proof}
We are now ready to prove part (c) of 
\cref{prop:discrete_step}.
\begin{corollary} \label{cor:discrete_iteration_conclusion_step}
Conditioned on $\EventUpdate$,
	\begin{align*}
	\bE[f(X')+f(Y')] - [f(X)+f(Y)]
	\geq{} &
	2(1 - 15\eps) \cdot [f(OPT(X, Y)) -\bE [f(OPT(X', Y'))]]  \\&- 2\eps^2\ln^{-1}(\eps^{-1}) \cdot \sum_{u \in Y \setminus X} (a_u + b_u)
	\enspace.
	\end{align*}
\end{corollary}
\begin{proof}
	Our first objective is to show that for every element $u \in Y \setminus X$ we have
	\begin{equation} \label{eq:discrete_basic_inequality}
	2 \cdot \max\{b_u r_u, a_u(1 - r_u)\} \leq a_u r_u + b_u(1 - r_u)
	\enspace.
	\end{equation}
	There are three cases that we need to consider. If $ a_u>0 $ and $ b_u>0 $, then Inequality~\eqref{eq:discrete_basic_inequality} 
	reduces to \[
	\frac{2a_ub_u}{a_u+b_u}\le \frac{a_u^2+b_u^2}{a_u+b_u}\enspace,
	\]
	which holds since $ a_u^2+b_u^2\ge 2a_ub_u $ for every two real numbers $a_u$ and $b_u$. Consider now the case in which $ a_u>0 $ and $ b_u\le 0$. In this case $r_u = 1$ and Inequality~\eqref{eq:discrete_basic_inequality} reduces to $ 2\max\{ b_u,0 \}\le a_u $, which is true because 
	$ 2\max\{ b_u,0 \} = 0 \le a_u $. Finally, consider the case of $ a_u\le 0 $. In this case $r_u = 0$ and Inequality~\eqref{eq:discrete_basic_inequality} reduces to $ 2\max\{ 0,a_u \}\le b_u $. To see why this inequality holds, note that since $ 
	a_u\le 0 $ and $ a_u+b_u \ge 0 $, we have $ b_u\ge 0 $ in this case, and thus, $ 2\max\{ 0,a_u \}=0\le b_u $. This completes the proof that Inequality~\eqref{eq:discrete_basic_inequality} always holds.
	
We are now ready to prove the corollary. Lemma~\ref{lem:discrete_loss_two_sided} applies conditioned on every given choice for $\delta$. By taking the expectation of its guarantee over the distribution of $\delta$ values resulting from the condition of the corollary (the condition is independent of the random decision made by the algorithm after choosing $\delta$), we get
\[
	f(OPT(X,Y))-\bE[f(OPT(X',Y'))]
	\leq
	\bE[\delta] \cdot \sum_{u \in U^+} \max\{b_u r_u, a_u(1 - r_u)\}
	\enspace.
\]
Similarly, Lemma~\ref{lem:discrete_gain_two_sided} applies conditioned on any fixed choice for the estimates $\{G(\delta') \mid \delta' \in \Gamma\}$ which obeys the condition of the current corollary. Thus, again we can take expectation of both its sides over the distribution of these estimates conditioned on the condition of the corollary, which gives us
\begin{align*}
	\bE[f(X') + f(Y')] - [f(X&) + f(Y)]\\
	\geq{} & (1 - 15\eps)\bE[\delta]\cdot \sum_{u \in U^+}[a_u r_u + 
	b_u(1 - r_u)] - 2\eps^2\ln^{-1}(\eps^{-1}) \cdot \sum_{u \in Y \setminus X} (a_u + b_u)
	\enspace.
\end{align*}
	
The corollary now follows by combining the last two inequalities with the inequality \[2 \cdot \sum_{u \in U^+}\max\{b_u r_u, a_u(1 - r_u)\} \leq \sum_{u \in U^+}[a_u r_u + b_u(1 - r_u)] \enspace,\] which is an immediate consequence of Inequality~\eqref{eq:discrete_basic_inequality}.
	%
\end{proof}

To complete the proof of Proposition~\ref{prop:discrete_step}, it remains to bound its adaptivity and the number of oracle queries it uses, which is done by the next lemma.

\begin{lemma}\label{lem:adaptivity_discrete_update}
	{\DiscreteUpdate} has constant 
	adaptivity and uses 
	$O(n\eps^{-3}\ln^2(\eps^{-1}))$
	 value oracle queries to $f$.
\end{lemma}
\begin{proof}
	Observe that all the value oracle queries used by {\DiscreteUpdate} can be 
	made in 
	two parallel steps. One step for computing $a$ and $b$, which requires $4n$ 
	value oracle queries to $f$, and one additional step for calculating the estimates $\{G(\delta') \mid \delta' \in \Gamma\}$. This already proves that {\DiscreteUpdate} has 
	constant 
	adaptivity. To prove the other part of the lemma, we still need to bound the number of value oracle queries used for calculating the estimates.
	
	 Since {\DiscreteUpdate} calculates one estimate for every value $\delta' 
	 \in \Gamma$, we get that it calculates only $O(\eps^{-1}\ln \eps^{-1})$ estimates by \cref{lem:bound_size_gamma}. Every estimate is 
	 done by averaging $m$ samples, 
	 and every sample is obtained using up to $4n$ value oracle queries, and 
	 thus, the number of queries requires for all the estimates is at most
	\begin{align*}
		4nm \cdot O(\eps^{-1}\ln \eps^{-1})
		={} &
		\lceil 
		\eps^{-2}\ln(112\eps^{-3}\ln^2(\eps^{-1}))/2 \rceil \cdot O(n\eps^{-1}\ln 
		\eps^{-1})\\
		\leq{} &
		\lceil \eps^{-2}\ln(112\eps^{-3} \cdot (\eps^{-1})^2)/2 \rceil \cdot O(n\eps^{-1}\ln 
		\eps^{-1})\\
		={} &
		O(\eps^{-2} \ln \eps^{-1}) \cdot O(n\eps^{-1}\ln \eps^{-1})
		=
		O(n\eps^{-3}\ln^2(\eps^{-1}))
		\enspace.
		\qedhere
	\end{align*}
\end{proof}

\subsection{The Procedure \DiscretePreProcess}\label{sub:discrete_preprocess}

In this section we describe the procedure {\DiscretePreProcess} used by Algorithm~\ref{alg:discrete_set}, and prove that it obeys all the properties guaranteed by Proposition~\ref{prop:discrete_pre-process}.  We remind the reader that this procedure is a counterpart of the procedure {\PreProcess} from Section~\ref{ssc:pre-process}. The pseudo-code of {\DiscretePreProcess} appears as Algorithm~\ref{alg:discrete_pre-process}. Given a value $\delta' \in [0, 1/2]$, we use in this procedure and its analysis the notation $\RSet(\delta')$ to denote a random pair of sets $(X', Y')$ whose distribution is defined as follows. For every element $u \in \cN$, independently, with probability $\delta'$ the element $u$ belongs to both sets $X'$ and $Y'$, with probability $\delta'$ it belongs to neither set, and with the remaining probability it belongs to $Y'$ but not to $X'$. Note that if $(X', Y') = \RSet(\delta')$, then $X' \subseteq Y'$ by definition.

\begin{algorithm}[th]
	\DontPrintSemicolon
	\caption{$\DiscretePreProcess(\tau)$} \label{alg:discrete_pre-process}
	Let $\Gamma \gets \{\delta' \in 
	[\eps, \nicefrac{1}{2}) \mid \exists_{j \in 
		\bZ, j \geq 1}\; \delta = \eps j \}$.\\
	For every $\delta' \in \Gamma$, consider the random expression
	\[
		\sum_{u\in \cN} [f(u\mid R^u_x(\delta') - u) - f(u \mid R^u_y(\delta') - u)]
		\enspace,
	\]
	where $R^u_x(\delta')$ and $R^u_y(\delta')$ are random variables such that $(R^u_x(\delta'), R^u_y(\delta')) \sim \RSet(\delta')$ for every element $u \in \cN$, independently. Let $G(\delta')$ be an estimate of 
	the expectation of this expression obtained by averaging $ m = \lceil 36\eps^{-2}\ln(3\eps^{-2}) \rceil $ 
	independent samples from the distribution of the expression.\\
	Let $\delta$ be the minimum value in $\Gamma$ for which $G(\delta) \le 30\tau $. If there is no such value, we set $ 
	\delta \gets \nicefrac{1}{2}$.\label{ln:delta_x_condition}\\

	\Return{$\RSet(\delta)$}.
\end{algorithm}

Let us now recall Proposition~\ref{prop:discrete_pre-process}. We remind the reader that our objective is to prove that Algorithm~\ref{alg:discrete_pre-process} obeys all the properties guaranteed by this proposition.

\begin{repproposition}{prop:discrete_pre-process}
\propDiscretePreProcess
\end{repproposition}

We observe that part~\eqref{item:discrete_initial_subset} of 
Proposition~\ref{prop:discrete_pre-process} follows immediately from the fact 
that Algorithm~\ref{alg:discrete_pre-process} returns a random set from the 
distribution of $\RSet(\delta)$. Proving the other parts of the proposition is 
more involved, and we begin with the following technical observation. Starting 
from this point, all the claims we present implicitly assume that the conditions 
of Proposition~\ref{prop:discrete_pre-process} hold.

\begin{observation}\label{obs:bound_Delta_xy}
	The set $ \Gamma $ is of size at most $ (2\eps)^{-1} $.
\end{observation}
\begin{proof}
The set $\Gamma$ contains only values from the range $[\eps, \nicefrac{1}{2}) $ that are of the 
	form $ \eps j $. The number of such values within this range is at most
	\[
	\left \lfloor \frac{1/2}{\eps} \right\rfloor
	\leq
	(2\eps)^{-1}\enspace.
	\qedhere
	\]
\end{proof}

Let us now define the event $\EventPreProcess$ references by Proposition~\ref{prop:discrete_pre-process}. For every $\delta' \in \Gamma$, we denote by $\bE_{\delta'}$ the expectation that $G(\delta')$ estimates. Then, the event $\EventPreProcess$ is the event that for every $\delta' \in \Gamma$ the estimate $G(\delta')$ is equal $\bE_{\delta'}$ up to an error of $\max\{\bE_{\delta'}/2, f(OPT)\}$.
\begin{lemma}\label{lem:probability_event_preprocess}
	$ \Pr[\EventPreProcess] \ge 1-\eps/3 $.
\end{lemma}
\begin{proof}
Algorithm~\ref{alg:discrete_pre-process} makes $|\Gamma| \leq (2\eps)^{-1}$ estimates (where the inequality is due to \cref{obs:bound_Delta_xy}). Thus, by the union bound, to prove that all the estimates are correct up to the error allowed by $\EventPreProcess$, it suffices to prove that every particular estimate is correct up to this error with probability at least $1 - 2\eps^2/3$. In the rest of the prove we show that this is indeed the case.

Fix then some value $\delta' \in \Gamma$, and let $ W_{i, u}$ be the value of the $i$-th sample of $f(u\mid R^u_x(\delta') - u) - f(u \mid R^u_y(\delta') - u)$ used to calculate $G(\delta')$. Observe that since $R^u_x(\delta')$ is always a subset of $R^u_y(\delta')$ the submodularity of $f$ guarantees that
\[
	W_{i, u}
	=
	f(u\mid R^u_x(\delta') - u) - f(u \mid R^u_y(\delta') - u)
	\geq
	0
	\enspace.
\]
In contrast, the submodularity of $f$ also guarantees that
\[
	W_{i, u}
	\leq
	f(u \mid \varnothing) -f(u \mid \cN - u)
	\leq
	f(\{u\}) + f(\cN - u)
	\leq
	2f(OPT)
	\enspace,
\]
where the second inequality holds by the monotonicity of $f$ and the last inequality follows from the observation that $\{u\}$ and $\cN - u$ are feasible solutions.

Let us now define $Z_{i, u} = W_{i, u} / [2f(OPT)]$ for every $1 \leq i \leq m$ and $u \in \cN$. Clearly these variables are independent, and moreover, they all belong to the range $[0, 1]^\cN$. Using this notation, the event that $G(\delta')$ is correct up to an error of $\max\{\bE_{\delta'}/2, f(OPT)\}$ can be written as the event that the random variable
\[
	Z = \sum_{i = 1}^m \sum_{u \in \cN} Z_{i, u}
\]
belongs to the range $[\min\{\bE[Z]/2, \bE[Z]-m/2\}, \max\{1.5 \cdot \bE[Z], \bE[Z] + m/2\}]$. In the rest of the proof we show that the last event happens with probability at most $2\eps^2/3$. There are two cases to consider. If $\bE[Z] \leq m\eps^2/3$, then the inequality $Z \geq \bE[Z] - m/2$ holds trivially because $Z$ is non-negative, and by Markov's inequality we also have
\[
	\Pr[Z \geq \bE[Z] + m/2]
	\leq
	\Pr[Z \geq m/2]
	\leq
	\frac{\bE[Z]}{m/2}
	\leq
	\frac{m\eps^2/3}{m/2}
	=
	\frac{2\eps^2}{3}
	\enspace.
\]
Consider now the case in which $\bE[Z] \geq m\eps^2/3$. Since $Z$ is the sum independent random variables taking values only from the range $[0, 1]$, the Chernoff bound guarantees that the probability that $Z \not \in [\bE[Z]/2, 1.5 \cdot \bE[Z] / 2]$ is at most
\[
	2\exp\left(-\frac{\bE[Z]}{12}\right)
	\leq
	2\exp\left(-\frac{m\eps^2}{36}\right)
	\leq
	2\exp(-\ln(3\eps^{-2}))
	=
	\frac{2\eps^2}{3}
	\enspace.
	\qedhere
\]
\end{proof}

The next lemma proves part~\eqref{item:discrete_potential_start} of Proposition~\ref{prop:discrete_pre-process}.

\begin{lemma}
Conditioned on $\EventPreProcess$,
	$\bE\left\{\sum_{u \in Y \setminus X} [f(u \mid X) - f(u \mid Y - u)]\right\} \leq 
		60\tau
		$.
\end{lemma}
\begin{proof}
We will prove the lemma conditioned on any fixed choice of values for the estimates $\{G(\delta') \mid \delta' \in \delta' \in \Gamma\}$ which make the $\EventPreProcess$ hold. One can observe that this will imply the lemma holds also unconditionally by the law of total expectation since the event $\EventUpdate$ is completely determined by the values of these estimates.

The fact that we fixed the estimates $\{G(\delta') \mid \delta' \in \delta' \in \Gamma\}$ implies that $\delta$ is also deterministic, and thus, we only need to handle the randomness used to construct $X$ and $Y$. If $\delta = 1/2$, then we deterministically have $X = Y$, which makes the left side of the inequality we need to prove $0$. Hence, the inequality holds in this case since its right hand side is non-negative.

In the rest of the proof we consider the case of $\delta < 1/2$. Note that the distributions of $X$ and $Y$ are identical to the distributions of $R^u_x(\delta)$ and $R^u_y(\delta)$, respectively. Thus, we get in this case
\begin{align*}
	\bE\mspace{200mu}&\mspace{-200mu}\left\{\sum_{u \in \cN} [f(u \mid X - u) - f(u \mid Y - u)]\right\}
	=
	\sum_{u \in \cN} \{\bE[f(u \mid X - u)] - \bE[f(u \mid Y - u)]\}\\
	={} &
	\sum_{u \in \cN} \{\bE[f(u \mid R^u_x(\delta) - u)] - \bE[f(u \mid R^u_y(\delta) - u)]\}\\
	\leq{} &
	\max\{G(\delta) + f(OPT), 2G(\delta)\}
	\leq
	\max\{30\tau + f(OPT), 60\tau\}
	=
	60 \tau
	\enspace,
\end{align*}
where the first inequality follows from our assumption that the value of $G(\delta)$ obey the requirement of $\EventPreProcess$, and thus, is equal to $\bE_\delta$ up to an error of $\max\{\bE_\delta/2, f(OPT)\}$, the second inequality follows from the way $\delta$ is chosen and the last equality holds since $\tau \geq f(OPT) / 5$.

To complete the proof of the lemma, it remains to observe that since $X \subseteq Y$, the submodularity of $f$ guarantees that $f(u \mid X - u) - f(u \mid Y - u) \geq 0$ for every element $u \in \cN$, and thus, it deterministically holds that
\[
	\sum_{u \in Y \setminus X} [f(u \mid X) - f(u \mid Y - u)]
	\leq
	\sum_{u \in \cN} [f(u \mid X) - f(u \mid Y - u)]
	\enspace.
	\qedhere
\]
\end{proof}

Our next objective is to prove part~\eqref{item:discrete_gain_loss_pre} of Proposition~\ref{prop:discrete_pre-process}. To do that, we first bound in Lemma~\ref{lem:discrete_loss_pre} the value of the expectation $\bE[f(OPT(X, Y))]$ appearing in this part of the proposition. The following is a know lemma that we use in the proof of  Lemma~\ref{lem:discrete_loss_pre}.
\begin{lemma}[Lemma~2.3 of~\cite{FMV11} (rephrased)] \label{lem:sampling_pair}
Let $f\colon 2^\cN \to \bR$ be submodular and $A, B \subseteq \cN$ two (not necessarily disjoint) sets. If $A(p)$ and $B(q)$ are independent random subsets of $A$ and $B$, respectively, such that $A(p)$ includes every element of $A$ with probability $p$ and $B(q)$ includes every element of $B$ with probability $q$, then
\[
	\bE[f(A(p) \cup B(q))]
	\geq
	(1 - p)(1 - q) \cdot f(\varnothing) + p(1-q) \cdot f(A) + q(1 - p) \cdot f(B) + pq \cdot f(A \cup B)
	\enspace.
\]
\end{lemma}

\begin{lemma} \label{lem:discrete_loss_pre}
Conditioned on a fixed choice of $\delta$, $\bE[f(OPT(X, Y))] \geq (1 - 2\delta) \cdot f(OPT)$.
\end{lemma}
\begin{proof}
Every element $u \in OPT$ belongs to $OPT(X, Y)$ with probability $1 - \delta$ and every element $u \not \in OPT$ belongs to $OPT(X, Y)$ with probability $\delta$. Moreover, for every element its membership in $OPT(X, Y)$ is independent from the membership of other elements in this set once $\delta$ is fixed. Thus, $OPT(X, Y)$ can be view as the union of two independent random sets, one of which is a random subset of $OPT$ containing every element of $OPT$ with probability $1 - \delta$, and the other is a random subset of $\cN \setminus OPT$ containing every element of $\cN \setminus OPT$ with probability $\delta$. By Lemma~\ref{lem:sampling_pair}, the expected value according to $f$ of such a set is at least
\begin{align*}
	\delta(1 - \delta) \cdot f(\varnothing) + (1 - \delta)^2 \cdot f(OPT)+ \delta^2 \cdot f(\cN \setminus OPT) &{}+ \delta(1 - \delta) \cdot f(\cN)\\
	\geq{} &
	(1 - \delta)^2 \cdot f(OPT)
	\geq
	(1 - 2\delta) \cdot f(OPT)
	\enspace,
\end{align*}
where the inequalities follows from the non-negativity of $f$.
\end{proof}


Next, we prove a lower bound on the expectation $\bE[f(X) + f(Y)]$, which is another term appearing in part~\eqref{item:discrete_gain_loss_pre} of Proposition~\ref{prop:discrete_pre-process}.
\begin{lemma} \label{lem:discrete_gain_pre}
	Conditioned on any fixed choice for the estimates $\{G(\delta') \mid \delta \in \Gamma\}$ which makes the event $\EventPreProcess$ happen, $\bE[f(X) + f(Y)] \geq 4(\delta - \eps) \cdot f(OPT)$.
\end{lemma}
\begin{proof}
If $\delta = \eps$, then the left side of the inequality that we need to prove is non-negative and its right side is $0$, which makes the lemma trivial. Thus, we may assume in the rest of the proof that $\delta \geq 2\eps$. Let us  define now by $R_1$ a random set that contains every element $u \in \cN$ with probability $\delta$, independently, and by $R_2$ a random set that contains every element of $R_1$ with probability $(\delta - \eps) / \delta$, independently. Note that $R_2$ is always a subset of $R_1$ and the two sets have the same distributions as $R^u_x(\delta - \eps)$ and $R^u_x(\delta)$, respectively (for every element $u \in \cN$). Using these observations we now get
\begin{align*}
	\bE[f(X)]
	={} &
	\bE[f(R^u_x(\delta))]
	=
	\bE[f(R_1)]
	=
	\bE[f(R_2)] + \bE[f(R_1 \mid R_2)]\\
	\geq{} &
	\bE\left[ \sum_{u \in \cN} \characteristic[u \in R_2] \cdot f(u \mid R_2 - u)\right] + \bE\left[ \sum_{u \in \cN} \characteristic[u \in R_1 \setminus R_2] \cdot f(u \mid \RSet(1/2) - u)\right]\\
	={} &
	\sum_{u \in \cN} \Pr[u \in R_2] \cdot \bE[f(u \mid R_2 - u)] + \sum_{u \in \cN} \Pr[u \in R_1 \setminus R_2] \cdot \bE[f(u \mid \RSet(1/2) - u)]\\
	={} &
	(\delta - \eps) \cdot \sum_{u \in \cN} \bE[f(u \mid R^u_x(\delta - \eps) - u)] + \delta \cdot \left(1 - \frac{\eps}{\delta}\right) \cdot \sum_{u \in \cN} \bE[f(u \mid \RSet(1/2) - u)]
	\enspace,
\end{align*}
where the inequality follows from the non-negativity and submodularity of $f$ and the penultimate equality holds since the distribution of $R_2 - u$ is independent of the membership of $u$ in $R_2$. Using an analogous argument we can also get
\[
	\bE[f(Y)]
	\geq
	(\delta - \eps) \cdot \sum_{u \in \cN} \bE[-f(u \mid R^u_y(\delta - \eps) - u)] + \delta \cdot \left(1 - \frac{\eps}{\delta}\right) \cdot \sum_{u \in \cN} \bE[-f(u \mid \RSet(1/2) - u)]
	\enspace.
\]
Adding this inequality to the previous one yields
\begin{align*}
	\bE[f(X) + f(Y)]
	\geq{} &
	(\delta - \eps) \cdot \sum_{u \in \cN} \bE[f(u \mid R^u_x(\delta - \eps) - u)-f(u \mid R^u_y(\delta - \eps) - u)]\\
	\geq{} &
	(\delta - \eps) \cdot \max\{G(\delta - \eps) - f(OPT), (2/3) \cdot G(\delta - \eps)\}\\
	\geq{} &
	(\delta - \eps) \cdot \max\{30\tau - f(OPT), 20\tau\}
	\geq
	(\delta - \eps) \cdot 4f(OPT)
	\enspace,
\end{align*}
where the second inequality follows from our assumption that the value of $G(\delta - \eps)$ obeys the event $\EventPreProcess$, the third inequality follows from the choice of $\delta$ and the last inequality follows from the assumption that $\tau \geq f(OPT) / 5$.
\end{proof}

Combining the last two lemmata, we can now get 
part~\eqref{item:discrete_gain_loss_pre} of 
Proposition~\ref{prop:discrete_pre-process}.
\begin{corollary}
	Conditioned on $ \EventPreProcess $, $\bE[f(X)+f(Y)] 
	\geq 
	2[f(OPT) - \bE[f(OPT(X, Y))]] - 4\eps \cdot f(OPT)$.
\end{corollary}
\begin{proof}
Lemma~\ref{lem:discrete_loss_pre} applies conditioned on every given choice for $\delta$. By taking the expectation of its guarantee over the distribution of $\delta$ values resulting from the condition of the corollary (the condition is independent of the random decision made by the algorithm after choosing $\delta$), we get
\[
	\bE[f(OPT(X,Y))]
	\geq
	(1 - 2\bE[\delta]) \cdot f(OPT)
	\enspace.
\]
Similarly, Lemma~\ref{lem:discrete_gain_pre} applies conditioned on any fixed choice for the estimates $\{G(\delta') \mid \delta' \in \Gamma\}$ which obeys the condition of the current corollary. Thus, again we can take expectation of both its sides over the distribution of these estimates conditioned on the condition of the corollary, which gives us
\[
	\bE[f(X) + f(Y)] \geq 4(\bE[\delta] - \eps) \cdot f(OPT)
	\enspace.
\]
The corollary now follows by combining the last two inequalities.
\end{proof}
\begin{lemma}\label{lem:adaptivity_discrete_pre-process}
	Algorithm~\ref{alg:discrete_pre-process} has constant adaptivity and uses 
	$O(n\eps^{-3}\ln(\eps^{-1}))$ value oracle queries to $f$.
\end{lemma}
\begin{proof}
	Observe that all the value oracle queries used by {\DiscretePreProcess} can 
	be made 
	in one step, and thus, the this procedure has constant adaptivity. To prove the other part of the lemma, we observe that {\DiscretePreProcess} uses oracle queries only for calculating the estimates $\{G(\delta') \mid \delta' \in \Gamma\}$. There are $|\Gamma|$ such estimates, each estimate is calculated based on $m$ samples and each sample requires $4n$ oracle queries. Thus, the total number of oracle queries used by {\DiscretePreProcess} is
	\[
		|\Gamma| \cdot m \cdot 4n
		\leq
		\frac{1}{2\eps} \cdot (36\eps^{-2}\ln(3\eps^{-2}) + 1) \cdot 4n
		=
		O(n\eps^{-3} \ln(\eps^{-1}))
		\enspace,
	\]
	where the first inequality follows from Observation~\ref{obs:bound_Delta_xy}.
\end{proof}

%% file: AlgorithmDRSubmodular.tex
\section{Algorithm for DR-Submodular Functions} \label{sec:dr_submodular}

In this section, we study the maximization of a DR-submodular function subject to a 
box constraint. A function $ F\colon\prod_{u \in \cN} [a_u,b_u]\to \mathbb{R} $ is 
DR-submodular~\cite{bian2016guaranteed} if for every two vectors $x,y\in \prod_{u \in \cN}
[a_u,b_u] $ and two scalars $u \in \cN $ and $ k\ge 0 $, it holds that
\[
F(ke_u+x)-F(x)\ge F(ke_u+y)-F(y)
\enspace,
\]
provided that $ x\le y $ and $ ke_u+x,ke_u+y\in \prod_{u \in \cN} [a_u,b_u] $, 
where $ e_u $ is the $ n $-dimensional vector whose $ u $-th component is $1$ and all its other components are zeros. One can note that the gradient of a 
differentiable DR-submodular function $F$ obeys $ \nabla F(x)\ge \nabla F(y) $ for every two vectors $ x,y \in 
\prod_{u \in \cN}
[a_u,b_u]$ such that $ x\le y $~\cite{bian2016guaranteed}. Another immediate observation is that every 
such function can be rescaled into a DR-submodular function defined on 
the hypercube $ [0,1]^\cN $. To be precise, $ G(x) = F(a+(b-a)x) $ is a 
DR-submodular function on the hypercube provided that $ F $ is DR-submodular on 
$ \prod_{u \in \cN} [a_u,b_u] $, and the two functions share their maximum values. Therefore, throughout this section we assume, without loss 
of generality, that $ F $ is DR-submodular on the hypercube. 
Additionally, we assume that $ F $ is non-negative and differentiable and that we have access to oracles that return $F(x)$ given a vector $x \in [0, 1]^\cN$ and return $\partial_u F(x)$ given such a vector and an element $u \in \cN$. Under these assumptions, we study 
the optimization problem $\max_{x\in [0,1]^\cN} F(x)$, which we refer to as {\UDRSM} in the remainder of this section.

Let $ \opt $ denote an arbitrary optimal solution for {\UDRSM}. Throughout this section, $ 
OPT(x,y) $ denotes $ (\opt \vee x)\wedge y $ for any two vectors $ x,y\in [0,1]^\cN $. We 
remark that this definition is different from the definition of $OPT(x, y)$ used in \cref{sec:algorithm}. However, the modification is a natural consequnce of the fact that the optimal solution is no longer a discrete set.

The result that we prove in this section is summarized in 
\cref{thm:DR_main_result_actual}.

\begin{theorem} \label{thm:DR_main_result_actual}
	For every constant $\eps\in (0,\nicefrac{1}{3})$, there is an algorithm 
	that assumes value 
	oracle access to the objective function $F$ 
	and achieves $(\nicefrac{1}{2} - 44\eps)$-approximation for {\UDRSM} using 
	$O(\eps^{-1})$ adaptive rounds and $O(n\eps^{-2}\log \eps^{-1})$ value oracle 
	queries to $F$.
\end{theorem}

The proposed algorithm that obeys all properties and guarantees promised in 
\cref{thm:DR_main_result_actual} is presented as 
\cref{alg:box_constrained_DRSM}. The algorithm is identical to 
\cref{alg:multilinear_oracle} except for its returned value. 
Since {\USM} anticipates a discrete solution,
\cref{alg:multilinear_oracle} returns a subset obtained by randomly rounding the 
potentially fractional vector $ x^i $. In contrast, 
\cref{alg:box_constrained_DRSM} returns the vector $ x^i $ itself. The two 
subroutines that \cref{alg:box_constrained_DRSM} invokes are exactly 
{\PreProcess} (\cref{alg:pre-process}) and {\Update} (\cref{alg:update}) from Section~\ref{sec:algorithm}.

\begin{algorithm}[ht]
	\caption{\texttt{Algorithm for Box-Constrained DR-Submodular Maximization}} 
	\label{alg:box_constrained_DRSM}
	\DontPrintSemicolon
	Let $\tau \gets F(\nicefrac{1}{2} \cdot \characteristic_\cN)$ and $\gamma 
	\gets 4\eps\tau$.\\
	Let $i \gets 0$ and $(x^0, y^0, \Delta^0) \gets \PreProcess( \tau)$.\\
	\While{$\Delta^i > 0$}
	{
		Let $(x^{i + 1}, y^{i + 1}, \Delta^{i + 1}) \gets \Update(x^i, y^i, 
		\Delta^i, \gamma)$.\\
		Update $i \gets i + 1$.
	}
	\Return{$x^i$}.
\end{algorithm}

The analysis of the number of adaptive rounds of \cref{alg:box_constrained_DRSM} and the number of oracle queries it uses is 
the same as the one done in the proof of \cref{thm:main_result_actual}. Thus, in the sequel, we only prove the 
first part of \cref{thm:DR_main_result_actual}, \ie, 
that \cref{alg:box_constrained_DRSM} achieves $(\nicefrac{1}{2} - 
13\eps)$-approximation for {\UDRSM} for every constant $\eps\in 
(0,\nicefrac{1}{25}]$.

In Section~\ref{sub:DR_step} we show that Proposition~\ref{prop:step}, which gives the properties of the procedure {\Update} under the {\USM} setting, applies (as is) also to the {\UDRSM} setting. For convenience, we repeat the proposition itself below.
\begin{repproposition}{prop:step}
	\stepProp
\end{repproposition}

The properties of the subroutine {\PreProcess} are given by \cref{prop:DR_pre-process}, which is proved in Section~\ref{sub:DR_pre-process}. One can observe that this proposition is identical to Proposition~\ref{prop:pre-process}, which gives the properties of {\PreProcess} under the {\USM} setting, except for two changes. First, that every appearance of $f(OPT)$ in Proposition~\ref{prop:pre-process} is replaced with an appearance of $F(\opt)$ in Proposition~\ref{prop:DR_pre-process}; and second, that Proposition~\ref{prop:DR_pre-process} refers also to the derivative oracle (which is not necessary in the {\USM} setting since the derivatives of the multilinear extension can be calculated using a value oracle).

\newcommand{\DRpreProcessProp}[1][]{The input for {\PreProcess} consists of a 
single value $\tau \geq 0$. If $\tau \geq F(\opt) / 4$, then {\PreProcess} 
outputs two vectors $x, y \in [0, 1]^\cN$ and a scalar $\Delta \in [0, 1]$ 
obeying
	\begin{compactenum}[(a)]
		\item $y - x = \Delta \cdot \characteristic_\cN$, \ifx&#1& \else 
		\label{item:DR_diff_delta_pre} \fi
		\item either $\Delta = 0$ or $\characteristic_\cN [\nabla F(x) - \nabla 
		F(y)] \leq 16\tau$ and \ifx&#1& \else \label{item:DR_potential_start} 
		\fi
		\item $F(x) + F(y) \geq 2[F(\opt) - F(OPT(x, y))] - 4\eps \cdot 
		F(\opt)$. 
		\ifx&#1& \else \label{item:DR_gain_loss_pre} \fi
	\end{compactenum}
	Moreover, {\PreProcess} requires only a constant number of adaptive rounds 
	and $O(n / \eps)$ value and derivative oracle queries to $F$.}
\begin{proposition} \label{prop:DR_pre-process}
	\DRpreProcessProp[l]
\end{proposition}

The following is a technical lemma that we use both in the analysis of Algorithm~\ref{alg:box_constrained_DRSM} and in the proof of Proposition~\ref{prop:DR_pre-process}.

\begin{lemma}\label{lem:DR_two_inequalities}
	For every vector $ z\in [0,1]^\cN $ and value $ \delta\in [0,1] $, the following two 
	inequalities hold: $ F(z\vee (\delta \cdot 
	\characteristic_\cN )) \ge (1-\delta) \cdot F(z) $ and $ 
	F(z\wedge((1-\delta) \cdot \characteristic_{\cN})) \ge (1-\delta) \cdot F(z) $.
\end{lemma}
\begin{proof}
Note that the lemma is trivial for $\delta = 0$. Thus, in the rest of the proof we assume $\delta \in (0, 1]$. Let us now prove the first inequality of the lemma. Consider the function $ G(t) = 
	F(z+\frac{t}{\delta} (\characteristic_\varnothing \vee (\delta \cdot \characteristic_{\cN}-z)) ) $. For every $ t\in [0,1] $, we have
	$ z+\frac{t}{\delta} (\characteristic_\varnothing\vee (\delta \cdot \characteristic_{\cN}-z)) \ge z \ge \characteristic_\varnothing $ 
	and 
	$ z+\frac{t}{\delta} (\characteristic_\varnothing\vee (\delta \cdot \characteristic_{\cN}-z)) \le 
	z+\frac{1}{\delta} (\characteristic_\varnothing\vee (\delta \cdot \characteristic_{\cN}-z)) = z \vee 
	(\characteristic_{\cN}+(1-1/\delta)z)
	\leq
	\characteristic_{\cN}
	$, where the last inequality holds since $ 1-1/\delta\le 0 $. Thus, the function $ G $ is well-defined for $ 
	t\in [0,1] $. By the chain rule, its derivative is
	\[
	\nabla F\left(z+\frac{t}{\delta} (\characteristic_\varnothing\vee 
	(\delta\characteristic_{\cN}-z)) \right) \cdot \frac{1}{\delta} (\characteristic_\varnothing\vee 
	(\delta\characteristic_{\cN}-z)).
	\]
	By the DR-submodularity of $ F $, the first factor in this derivative is non-decreasing in $ t $. 
	Since the second term is a non-negative constant vector, we obtain that the entire derivative is 
	non-decreasing in $ t $, which implies that $G$ is concave. Therefore, 
	\[
	F(z\vee \delta\characteristic_{\cN} )=G(\delta)\ge (1-\delta) \cdot G(0) + 
	\delta \cdot G(1) \ge (1-\delta) \cdot F(z),
	\]
	where the second inequality holds since $ G(0) = F(z) $ and $ G(1)\ge 0 $ (this follows from the non-negativity of $ F $).
	
	We now get to proving the second inequality of the lemma. Consider 
	the function $ H(t) = F(z-\frac{t}{\delta}(\characteristic_\varnothing \vee ( z - 
	(1-\delta)\characteristic_{\cN} ) ))$. For every $ 
	t\in [0,1] $, we have $ z-\frac{t}{\delta}(\characteristic_\varnothing\vee ( z - 
	(1-\delta)\characteristic_{\cN} ) ) \le z \le \characteristic_\cN $ and $ 
	z-\frac{t}{\delta}(\characteristic_\varnothing\vee ( z - 
	(1-\delta)\characteristic_{\cN} ) ) \ge z-\frac{1}{\delta}(\characteristic_\varnothing\vee ( z - 
	(1-\delta)\characteristic_{\cN} ) ) = z\wedge ( 
	\frac{1-\delta}{\delta}(1-z) ) \ge \characteristic_\varnothing $. Thus, the function $ H $ is well-defined for $ 
	t\in [0,1] $. Additionally, we can prove that $H$ is concave using an argument similar to the one used above to prove that $G$ is concave. Therefore,
	\[
	F(z\wedge (1-\delta)\characteristic_{\cN} ) = H(\delta) \ge 
	(1-\delta) \cdot H(0) + \delta \cdot H(1) \ge (1-\delta) \cdot F(z),
	\]
	where the second inequality holds since $ H(0)=F(z) $ and $ H(1)\ge 0 $.
\end{proof}

The following two lemmata correspond to lemmata from Section~\ref{sec:algorithm}.

\begin{lemma}[Corresponds to Lemma~\ref{lem:conditions_hold_two_sided}] \label{lem:DR_conditions_hold_two_sided}
	It always holds that $\tau \in [\nicefrac{1}{4}\cdot F(\opt), F(\opt)]$, 
	and 
	for every integer 
	$0 \leq i \leq \ell$ it holds that $x^i, y^i \in [0, 1]^\cN$, $\Delta^i \in 
	[0, 1]$ and $x^i + \Delta^i \cdot \characteristic_\cN = y^i$.
\end{lemma}
\begin{proof}
	Observe that
	\begin{align*}
	F(\nicefrac{1}{2}\cdot \characteristic_{\cN}) ={} & F((\opt \vee 
	(\nicefrac{1}{2}\cdot\characteristic_{\cN}) ) \wedge 
	((1-\nicefrac{1}{2})\cdot\characteristic_{\cN}) )\\\ge{} & (1-\nicefrac{1}{2}) \cdot
	F(\opt 
	\vee 
	\nicefrac{1}{2}\cdot \characteristic_{\cN})\ge (1-\nicefrac{1}{2})^2 \cdot
	F(\opt) = 
	\nicefrac{1}{4}\cdot F(\opt)
	\enspace,
	\end{align*}
	where the first inequality follows from the second inequality of 
	\cref{lem:DR_two_inequalities} by setting 
	$ z = \opt \vee (\nicefrac{1}{2}\cdot \characteristic_{\cN}) $ and the second 
	inequality 
	follows from the first inequality of \cref{lem:DR_two_inequalities} by setting $ 
	z = \opt $. Moreover, by the definition of $ \opt $, we have $ \tau = F(\nicefrac{1}{2} \cdot \characteristic_\cN) \le F(\opt) $, which completes the proof of the first part of the lemma, \ie, that $\tau \in [\nicefrac{1}{4} \cdot F(\opt), F(\opt)]$.
	
	The proof of the 
	second part of the lemma is identical to the proof of the corresponding part of \cref{lem:conditions_hold_two_sided}, and thus, we omit it.
\end{proof}

\begin{lemma}[Corresponds to Lemma~\ref{lem:initial_bound_two_sided}] \label{lem:DR_initial_bound_two_sided}
	$\ell \leq 5\eps^{-1}$ and $F(x^\ell) + F(y^\ell) \geq 2(1 - 3\eps) \cdot 
	[F(\opt) - F(OPT(x^\ell, y^\ell))] - 168\eps \cdot F(\opt)$.
\end{lemma}
\begin{proof}
The proof of Lemma~\ref{lem:initial_bound_two_sided} can be used to prove the current lemma under the following slight modifications.
\begin{itemize}
	\item Every appearance of $f(OPT)$ should replaced with $F(\opt)$.
	\item Every reference to Proposition~\ref{prop:pre-process} or Lemma~\ref{lem:conditions_hold_two_sided} should be replaced with a reference to Proposition~\ref{prop:DR_pre-process} or Lemma~\ref{lem:DR_conditions_hold_two_sided}, respectively.
	\item The proof uses the inequality $F(OPT(x^0, y^0)) \leq F(\opt)$. This inequality follows from the definition of $\opt$ since $OPT(x^0, y^0) \in [0, 1]^\cN$. \qedhere
\end{itemize}
\end{proof}

The following corollary completes the proof of the approximation guarantee part of 
\cref{thm:DR_main_result_actual}.

\begin{corollary}
	$
	F(x^\ell)
	\geq
	(\nicefrac{1}{2} - 44\eps) \cdot F(\opt)
	$.
	Hence, the approximation ratio of Algorithm~\ref{alg:multilinear_oracle} is 
	at least $\nicefrac{1}{2} - 44\eps$.
\end{corollary}
\begin{proof}
	

Observe that $\Delta^\ell = 0$ because otherwise the algorithm would not have stopped after $\ell$ iterations. Thus, $y^\ell = x^\ell + \Delta^\ell \cdot \characteristic_\cN = x^\ell$ and $OPT(x^\ell, y^\ell) = (\opt \vee x^\ell) \wedge y^\ell = x^\ell$. Plugging these observations into the guarantee of Lemma~\ref{lem:DR_initial_bound_two_sided}, we get
\[
	2F(x^\ell) \geq 2(1 - 3\eps) \cdot [F(\opt) - F(x^\ell)] - 168\eps \cdot F(\opt)
	\enspace,
\]
and the corollary now follows immediately by rearranging the last inequality and using the non-negativity of $F$.
\end{proof}

\subsection{The Procedure {\Update}}\label{sub:DR_step}

In this subsection, we prove that \cref{prop:step} still holds in the context of {\UDRSM}. Let us recall the statement of this proposition.

\begin{repproposition}{prop:step}
	\stepProp
\end{repproposition}

Proposition~\ref{prop:step} was proved in the context of {\USM} in Section~\ref{ssc:update}. Except for the proof of Lemma~\ref{lem:loss_two_sided}, all the proofs in this section can be made to apply to {\UDRSM} by simply replacing every reference to the submodularity of $f$ with a reference to the DR-submoduarlity of $F$. Thus, to prove that Proposition~\ref{prop:step} holds in the context of {\UDRSM}, it suffices to prove that Lemma~\ref{lem:loss_two_sided} holds in this setting, which is what we do next.

Recall that $U^+ \triangleq \{u \in \cN \mid a_u > 0 \text{ and } b_u > 0\}$.
\begin{replemma}{lem:loss_two_sided}
\lossTwoSided
\end{replemma}
\begin{proof}
	For every $ t\in [0,\delta] $, define $ v(t) \triangleq (\opt \vee (x + t r)) 
	\wedge (y - 
	t(\characteristic_\cN - r)) $, $ A(t)\triangleq \{ u\in \cN\mid \opt_u < 
	x_u + tr_u \} $, and $ B(t)\triangleq \{ u\in \cN\mid \opt_u > y_u-t(1-r_u) 
	\} $. Note that $ A(t) $ and $ B(t) $ are disjoint subsets of $ \cN $.
	Using the chain rule, we get
	\begin{align*}
	F(OPT(x, y)) - F(OPT(x', y')) 
	&=
	F((\opt \vee x) \wedge y) - F((\opt \vee 
	x') \wedge y') 
	=
	-\int_0^\delta \frac{d F(v(t))}{dt} dt\\
	={} &
	\int_0^\delta \left\{  \sum_{u\in B(t)} (1-r_u)\cdot \partial_u F(v(t)) - 
	\sum_{u\in A(t)} r_u \cdot \partial_u F(v(t)) \right\} dt
	\enspace.
	\end{align*}
	Using the DR-submodularity of $F$ and the fact that $x \leq 
	v(t) \leq y$, the rightmost side of the last equation can be upper bounded 
	as follows.
	\begin{align*}
	F(OPT(x, y)&) - F(OPT(x', y'))
	\leq
	\int_0^{\delta} \left\{\sum_{u \in B(t)} (1 - r_u) \cdot \partial_u F(x) - 
	\sum_{u \in A(t)}  r_u \cdot \partial_u F(y) 
	\right\} dt\\
	={} &
	\int_0^{\delta} \left\{\sum_{u \in B(t)} a_u(1 - r_u) + \sum_{u \in A(t)} 
	 b_u r_u \right\} dt
	\leq
	\int_0^{\delta} \sum_{u \in A(t) \cup B(t)} \mspace{-18mu} \max\{b_u r_u, a_u (1 - r_u)\} 
	dt\\
	\leq{} &
	\int_0^{\delta} \sum_{u \in \cN} \max\{b_u r_u, a_u (1 - r_u)\} 
	=
	\delta \cdot \sum_{u \in \cN} \max\{b_u r_u, a_u (1 - r_u)\} 
	\enspace,
	\end{align*}
	where the last inequality holds since for every element $u \in \cN$ the maximum $\max\{a_u, b_u\}$ is non-negative by Observation~\ref{obs:gain_bound}, and thus,  the maximum $\max\{b_u r_u, a_u(1 - r_u)\}$ is also non-negative.
	
	To complete the proof of the lemma, it remains to observe that for every element $u \in \cN \setminus U^+$ it holds that $\max\{b_u r_u, a_u(1 - r_u)\} = 0$. To see that this is the case, note that every such element $u$ must fall into one out of only two possible options. The first option is that $a_u > 0$ and $b_u \leq 0$, which imply $r_u = 1$, and thus, $\max\{b_u r_u, a_u(1 - r_u)\} = \max\{b_u, 0\} = 0$. The second option is that $a_u \leq 0$, which implies $r_u = 0$, and thus, $\max\{b_u r_u, a_u(1 - r_u)\} = \max\{0, a_u\} = 0$.
\end{proof}

\subsection{The Procedure {\PreProcess}}\label{sub:DR_pre-process}

In this subsection, we present the proof of \cref{prop:DR_pre-process}. Before 
presenting the proof, let us recall the statement of the proposition.

\begin{repproposition}{prop:DR_pre-process}
	\DRpreProcessProp
\end{repproposition}

We omit the proofs of parts~(\ref{item:diff_delta_pre}) and~(\ref{item:potential_start}) because they are identical to the proofs of the corresponding parts in Proposition~\ref{prop:pre-process} (these parts are proved by \cref{obs:diff_delta_pre,obs:potential_start}, respectively). Similarly, the analysis of the number adaptive rounds 
and the number of value oracle queries is identical to the proof of 
\cref{lem:pre-process_adaptivity}. Next, we prepare some results in order to 
prove part (\ref{item:DR_gain_loss_pre}) of Proposition~\ref{prop:DR_pre-process}.

\begin{lemma}[Counterpart of Lemma~\ref{lem:loss_pre}] \label{lem:DR_loss_pre}
	$F(OPT(x, y)) \geq \Delta \cdot f(\opt)$.
\end{lemma}
\begin{proof}
	 Observe that
	 \begin{align*}
	 F(OPT(x,y)) ={} & F((\opt \vee (\delta \cdot \characteristic_{\cN}) ) \wedge 
	 ((1-\delta) \cdot \characteristic_{\cN}) )\\\ge{} & (1-\delta) \cdot F(\opt \vee 
	 (\delta \cdot \characteristic_{\cN}))\ge (1-\delta)^2 \cdot F(\opt)\ge (1-2\delta) \cdot F(\opt)
	 \enspace,
	 \end{align*}
	 where the first inequality follows from the second inequality of 
	 \cref{lem:DR_two_inequalities} when setting 
	 $ z = \opt \vee (\delta \cdot \characteristic_{\cN}) $, the second inequality 
	 follows from the first inequality of \cref{lem:DR_two_inequalities} when setting $ 
	 z = \opt $, and the third 
	 inequality is due to the non-negativity of $ F $.
\end{proof}

\begin{lemma}[Counterpart of \cref{lem:gain_pre}] \label{lem:DR_gain_pre}
	$F(x) + F(y) \geq (\delta - \eps) \cdot 16\tau$.
\end{lemma}
\begin{proof}
	The proof is identical to that of \cref{lem:gain_pre}, except that $ 
	f(\varnothing) $ should be replaced by $ F(\characteristic_\varnothing) $.
\end{proof}

The following corollary of the last two lemmata completes the proof of part 
(\ref{item:DR_gain_loss_pre}).

\begin{corollary}[Counterpart of 
\cref{cor:lower_bound_Fx_Fy}]\label{cor:DR_lower_bound_Fx_Fy}
	If $\tau \geq F(\opt)/4$, then $F(x) + F(y) \geq 2[F(\opt) - F(OPT(x, y))] 
	- 4\eps \cdot F(\opt)$.
\end{corollary}
\begin{proof}
	The proof is identical to that of \cref{cor:lower_bound_Fx_Fy}, except that 
	$ f(OPT) $ should be replaced by $ F(\opt) $.
\end{proof}